\documentclass[10pt,journal,compsoc]{IEEEtran}

\usepackage[nocompress]{cite}
\usepackage{subfigure}
\usepackage{multirow}
\usepackage{graphicx}
\usepackage[space]{grffile}
\usepackage{mathtools, nccmath}
\usepackage{dblfloatfix}
\usepackage{amsmath}
\usepackage{amsmath}
\usepackage{booktabs}
\usepackage{dsfont}
\usepackage{wrapfig}
\usepackage{algorithm}
\usepackage{amsthm}
\usepackage{float}
\usepackage{bbm}
\usepackage{mathtools}
\usepackage{bbm}
\usepackage{algpseudocode}
\usepackage{mathtools}
\usepackage{amsmath,amssymb,amsfonts}
\usepackage{enumitem}
\usepackage{pifont}
\usepackage{xcolor}
\usepackage{stackengine}
\newcommand\xrowht[2][0]{\addstackgap[.5\dimexpr#2\relax]{\vphantom{#1}}}
\DeclareMathOperator*{\argmax}{arg\,max}
\DeclareMathOperator*{\argmin}{arg\,min}

\newtheorem{theorem}{Proposition}
\usepackage{bbding}
\usepackage{pifont}
\usepackage{wasysym}
\usepackage{amssymb}

\usepackage{multibib}
\newcites{MS}{References}

\begin{document}
\bstctlcite{IEEEexample:BSTcontrol}

\title{COSCO: Container Orchestration using Co-Simulation and Gradient Based Optimization for Fog Computing Environments}

\author{
        Shreshth~Tuli,
        Shivananda~Poojara,
        Satish~N.~Srirama,
        Giuliano~Casale
    and~Nicholas~R.~Jennings
\IEEEcompsocitemizethanks{
\IEEEcompsocthanksitem S. Tuli, G. Casale and N. R. Jennings are with the Department
of Computing, Imperial College London, United Kingdom. E-mails: \{s.tuli20, g.casale, n.jennings\}@imperial.ac.uk.\protect
\IEEEcompsocthanksitem S. Poojara is with the Institute of Computer Science, University of Tartu, Estonia. E-mail: poojara@ut.ee.\protect
\IEEEcompsocthanksitem S.~N.~Srirama is with the School of Computer and Information Sciences, University of Hyderabad, India. E-mail: satish.srirama@uohyd.ac.in.\protect
}
}


\markboth{Accepted in IEEE Transactions on Parallel and Distributed Systems}%
{Tuli \MakeLowercase{\textit{et al.}}: --- }

\IEEEtitleabstractindextext{%
\begin{abstract}
Intelligent task placement and management of tasks in large-scale fog platforms is challenging due to the highly volatile nature of modern workload applications and sensitive user requirements of low energy consumption and response time. Container orchestration platforms have emerged to alleviate this problem with prior art either using heuristics to quickly reach scheduling decisions or AI driven methods like reinforcement learning and evolutionary approaches to adapt to dynamic scenarios. The former often fail to quickly adapt in highly dynamic environments, whereas the latter have run-times that are slow enough to negatively impact response time. Therefore, there is a need for scheduling policies that are both reactive to work efficiently in volatile environments and have low scheduling overheads. To achieve this, we propose a Gradient Based Optimization Strategy using Back-propagation of gradients with respect to Input (GOBI). Further, we leverage the accuracy of predictive digital-twin models and simulation capabilities by developing a Coupled Simulation and Container Orchestration Framework (COSCO). Using this, we create a hybrid simulation driven decision approach, GOBI*, to optimize Quality of Service (QoS) parameters. Co-simulation and the back-propagation approaches allow these methods to adapt quickly in volatile environments.  Experiments conducted using real-world data on fog applications using the GOBI and GOBI* methods, show a significant improvement in terms of energy consumption, response time, Service Level Objective and scheduling time by up to 15, 40, 4, and 82 percent respectively when compared to the state-of-the-art algorithms.
\end{abstract}

}

\maketitle

\IEEEdisplaynontitleabstractindextext

\IEEEpeerreviewmaketitle

\IEEEraisesectionheading{\section{Introduction}\label{sec:introduction}}
\noindent
Fog computing is an emerging paradigm in distributed systems, which encompasses all intermediate devices between the Internet of Things (IoT) layer (geo-distributed sensors and actuators) and the cloud layer (remote cloud platforms). It can reduce latency of compute, network and storage services by placing them closer to end users leading to a multitude of benefits. However, fog environments pose several challenges when integrating with real-world applications. For instance, many applications, especially in healthcare, robotics and smart-cities demand ultra low response times specifically for applications sensitive to Service Level Objectives (SLO)~\cite{gill2019transformative}. Other applications involving energy scavenging edge devices and renewable resources need supreme energy efficiency in task execution~\cite{toor2019energy}. The challenge of reaching low response times and energy consumption is further complicated by modern-day application workloads being highly dynamic~\cite{tuli2020dynamic} and host machines having non-stationary resource capabilities~\cite{wang2020dyverse}.

To provide quick and energy efficient solutions, many prior works focus on developing intelligent policies to schedule compute tasks on fog hosts~\cite{basu2019learn, han2018fog, pham2017cost}. These methods have been dominated by heuristic techniques~\cite{skarlat2017optimized, beloglazov2012optimal, pham2017cost, choudhari2018prioritized}. Such approaches have low scheduling times and work well for general cases, but due to steady-state or stationarity assumptions, they provide poor performance in non-stationary heterogeneous environments with dynamic workloads~\cite{beloglazov2012optimal}. To address this, some prior approaches use more intelligent and adaptive schemes based on evolutionary methods and reinforcement learning. These methods adapt to changing scenarios and offer promising avenues for dynamic optimization~\cite{fox2019learning}. However, they too are unable to efficiently manage volatile fog environments because of their poor modelling accuracies and low scalability~\cite{basu2019learn, tuli2020dynamic}. For accurate and scalable modelling of the fog environment, there have been many works that use deep learning based local search or learning models with neural networks which approximate an objective function such as energy consumption or response time~\cite{tuli2020dynamic, liu2020collaborative, basu2019learn}. As these neural networks approximate objective functions of optimization problem, they are often referred to as ``neural approximators''~\cite{johnson2000accuracy, villarrubia2018artificial}. Specifically, many recent state-of-the-art techniques primarily use optimization methods like genetic algorithms (GA)~\cite{liu2020collaborative} and policy-gradient learning~\cite{tuli2020dynamic} to optimize QoS parameters, thanks to their generality. However, gradient-free methods like GA are slow to converge to optima due to undirected search schemes~\cite{rios2013derivative}. Further, policy-gradient learning takes time to adapt to sudden changes in the environment and shares the same problem of high scheduling overheads. Such high scheduling times limit the extent of the possible improvement of latency and subsequently SLO violations. This is not suitable for highly volatile environments where host and workload characteristics may suddenly and erratically change. Thus, there is a need for an approach which not only can adapt quickly in volatile environments but also have low scheduling overheads to efficiently handle modern workload demands.

To solve this problem, a natural choice is to use directed schemes like A* search or gradient-based optimization strategies. Even though such strategies have been shown to converge much faster than gradient-free approaches, prior work does not use them due to the highly non-linear search spaces in real-world problems which can cause such methods to get stuck in local optima~\cite{bogolubsky2016learning}. Moreover, prior works fail to leverage recent advances like root-mean-square propagation, and momentum and annealed gradient descent with restarts that help prevent the local optima problem~\cite{kingma2014adam,pan2015annealed,loshchilov2016sgdr}. Given this, we believe that by leveraging the strength of neural networks to accurately approximate QoS parameters~\cite{scarselli1998universal}, we can apply gradient-based algorithms in tandem with advances that reduce the likelihood of such methods getting stuck at local optima. Prior work also establishes that neural approximators are able to precisely model the gradients of the objective function with respect to input using back-propagation allowing us to use them in gradient-based methods for quick convergence~\cite{nguyen1999approximation}. When taken together, this suite of approaches should provide a quick and efficient optimization method. In particular, we formulate a gradient-based optimization algorithm (\textbf{GOBI}) which calculates the gradients of neural networks with respect to input for optimizing QoS parameters using advanced gradient-based optimization strategies. We establish by experiments that our approach provides a \textit{faster and more scalable optimization strategy} for fog scheduling compared to the state-of-the-art methods. 

However, simply using gradient-based optimization is not sufficient as data driven neural models can sometimes saturate~\cite{rakitianskaia2015measuring}. This is when feeding more data to the neural model does not improve the performance. Further optimization of QoS in such cases is hard and more intelligent schemes are required. 
Coupled-simulation (also referred as co-simulation or symbiotic simulation) and execution of tasks have been shown to be a promising approach for quickly obtaining an estimate of QoS parameters in the near future~\cite{gomes2018co, rehman2019cloud, onggo2018symbiotic}. Specifically, coupled-simulation allows us to run a simulator in the background with the scheduling algorithms to facilitate decision making. However, prior works use this to aid search methods and not to generate more data to facilitate decision making of an AI model. The latter requires development of a new interface between the scheduler and simulator. Hence, we develop \textbf{COSCO}: \textbf{Co}upled \textbf{S}imulation and \textbf{C}ontainer \textbf{O}rchestration Framework to leverage simulated results to yield better QoS. The COSCO framework is the \textit{first} to allow a single or multi-step simulation of container migration decisions in fog environments. It enables a scheduler to get an estimate of the QoS parameters at the end of a future interval for \textbf{better prediction and subsequently optimum schedules}. Container migration refers to the process of moving an application between different physical or virtual hosts and resuming computation on the target host. This allows us to run the GOBI approach, simulate the schedule (with a predicted workload model) and provide the objective values to another neural approximator which can then better approximate the objective function and lead to improved performance. We call this novel optimization loop \textbf{GOBI*} (Figure~\ref{fig:gobistar}). The \textit{interactive training dynamic} between GOBI and GOBI* aides the latter to converge quickly and adapt to volatile environments (more details in Section~\ref{sec:gobi2}). 

In summary, the key contributions of this paper are:%
\begin{itemize}[leftmargin=*]
    \item We present a novel framework, COSCO, which is the first to allow \textit{coupled simulation and container orchestration} in fog environments.
    \item We propose a gradient based back-propagation approach (GOBI) for \textit{fast and scalable scheduling} and show that it outperforms the state-of-the-art schedulers.
    \item We also propose an extended approach (GOBI*) that leverages simulation results from COSCO running GOBI's decision, to provide \textit{improved predictions and scheduling decisions} with lower scheduling overheads.
\end{itemize}
Validating both GOBI and GOBI* on physical setups with real-world benchmark data shows that, GOBI gives lower performance than GOBI*. However, GOBI is more suitable for resource constrained fog brokers, because of its low computational requirements. GOBI*, on the other hand, is more suitable for setups with critical QoS requirements and powerful fog brokers.

\begin{figure}[]
    \centering
    \includegraphics[width=0.95\columnwidth]{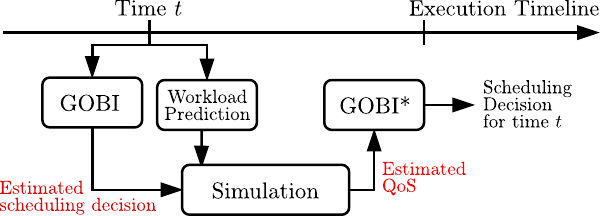}
    \caption{The GOBI* loop.}
    \label{fig:gobistar}
\end{figure}

The rest of the paper is organized as follows. Related work is overviewed in Section~\ref{sec:related-work}. Section~\ref{sec:motivation} provides a motivating example for the problem. Section~\ref{sec:system} describes the system model and formulates the problem specifications. Section~\ref{sec:gobi} details GOBI and Section~\ref{sec:cosco} introduces the COSCO framework. Section~\ref{sec:gobi2} presents the GOBI* approach, integrating GOBI with COSCO's co-simulation feature. A performance evaluation of the proposed methods is shown in Section~\ref{sec:perf_eval}. Finally, Section~\ref{sec:conclusion} concludes.

\begin{table*}[]
    \centering
    \caption{Comparison of Related Works with Different Parameters (\checkmark means that the corresponding feature is present).}
    \resizebox{\textwidth}{!}{
    \begin{tabular}{@{}lccccccccccc@{}}
    \toprule 
    \multirow{2}{*}{Work} & \multicolumn{1}{c}{Edge} & Heterogeneous & Coupled & Stochastic & Adaptive & \multirow{2}{*}{Method} & \multicolumn{5}{c}{Optimization Parameters}\tabularnewline
    \cline{8-12} 
     & Cloud & Environment & Simulation & Workload & QoS &  & Energy & Response Time & SLO Violations & Fairness & Scheduling Time\tabularnewline
    \midrule 
    \cite{beloglazov2012optimal, nashaat2019smart,samanta2019battle} &  &  &  & \checkmark &  & Heuristics & \checkmark & \checkmark & \checkmark &  & \tabularnewline

    \cite{wang2019empowering, han2018fog, wang2020effective, tuli2020ithermofog} & \checkmark & \checkmark &  & \checkmark &  & GA & \checkmark & \checkmark & \checkmark &  & \tabularnewline

    \cite{liu2020pond, krishnasamy2018augmenting} &  & \checkmark &  & \checkmark & \checkmark & MaxWeight &  & \checkmark &  &  & \tabularnewline

    \cite{zhang2018double, basu2019learn, gazori2019saving, tang2018migration} &  & \checkmark &  & & \checkmark & Deep RL & \checkmark & \checkmark &  &  & \checkmark\tabularnewline

    \cite{tuli2020dynamic, ghosal2020deep} & \checkmark & \checkmark &  & \checkmark & \checkmark & Policy Gradient & \checkmark & \checkmark & \checkmark & \checkmark & \tabularnewline

    \textbf{This work} & \checkmark & \checkmark & \checkmark & \checkmark & \checkmark & GOBI/GOBI{*} & \checkmark & \checkmark & \checkmark & \checkmark & \checkmark\tabularnewline
    \bottomrule 
    \end{tabular}}
    \label{tab:related_works}
\end{table*}

\section{Related Work}
\label{sec:related-work}

We now analyze  prior work in more detail (see Table~\ref{tab:related_works} for an overview).


\textbf{Heuristic methods:} Several studies~\cite{beloglazov2012optimal, nashaat2019smart, samanta2019battle} have shown how heuristic based approaches 
perform well in a variety of scenarios for optimization of diverse QoS parameters. Methods like \textit{LR-MMT}  schedules workloads dynamically based on local regression (LR) to predict which hosts might get overloaded in the near future and select containers/tasks which can be migrated in the minimum migration time (MMT). Both LR and MMT heuristics are used for overload detection and task selection, respectively~\cite{beloglazov2012optimal, nashaat2019smart}. Similarly, \textit{MAD-MC} calculates the median-attribute deviation (MAD) of the CPU utilization of the hosts and select containers based on the maximum correlation (MC) with other tasks~\cite{beloglazov2012optimal}. However,  such heuristic based approaches fail to model stochastic environments with dynamic workloads~\cite{tuli2020dynamic}. This is corroborated by our results in Section~\ref{sec:perf_eval}. To overcome this, GOBI and GOBI* are able to adapt to dynamic scenarios by constantly learning the mapping of scheduling decisions with objective values. This makes them both robust to diverse, heterogeneous and dynamic scenarios.

\textbf{Evolutionary models:} Other prior works have shown that evolutionary based methods, and generally gradient-free approaches, perform well in dynamic scenarios~\cite{wang2019empowering, han2018fog, wang2020effective, tuli2020ithermofog}. Evolutionary approaches like genetic algorithms (GA) lie in the domain of gradient-free optimization methods. The \textit{GA} method schedules workloads using a neural model to approximate the objective value (as in GOBI) and a genetic-algorithm to reach the optimal decision~\cite{han2018fog}.
However, gradient-free methods have a number of drawbacks. They take much longer to converge \cite{bogolubsky2016learning} and are not as scalable~\cite{rios2013derivative} as gradient-based methods. Moreover, non-local jumps can lead to significant change in the scheduling decision which can lead to a high number of migrations (as seen in Figures~\ref{fig:f_migration_time}-\ref{fig:s_migration_time}). 

\textbf{MaxWeight based schedulers:} Over the years, MaxWeight scheduling has gained popularity for its theoretical guarantees and ability to reduce resource contention~\cite{liu2020pond, krishnasamy2018augmenting}. However, MaxWeight policies can exhibit poor delay performance, instability in dynamic workloads and spatial inefficiency~\cite{bae2019beyond, van2009instability, van2013inefficiency}. We use the pessimistic-optimistic online dispatch approach, \textit{POND} by Liu et al.~\cite{liu2020pond} which is a variant of the MaxWeight approach~\cite{markakis2013max}. POND uses constrained online dispatch with unknown arrival and reward distributions. It uses virtual queues for each host to track violation counts and an Upper-Confidence Bound (UCB) policy~\cite{auer2002finite} to update the expected rewards of each allocation from the signal $\mathcal{O}(P_t)$. The final decision is made using the MaxWeight approach with the weights as the expected reward values. To minimize objective value in our setup, we provide the objective score as negative reward. Due to the inability of MaxWeight approaches to adapt to volatile scenarios, their wait times are high.  

\textbf{Reinforcement Learning models:} Recently, reinforcement learning based methods 
have shown themselves to be robust and versatile to diverse workload characteristics and complex edge setups~\cite{tuli2020dynamic, tang2018migration, basu2019learn, gazori2019saving}. Such methods use a Markovian assumption of state which is the scheduling decision at each interval. Based on new observations of reward signal (negative objective score in our setting) they explore or exploit their knowledge of the state-space to converge to an optimal decision. A recent method, \textit{DQLCM}, models the container migration problem as a multi-dimensional Markov Decision Process (MDP) and uses a deep-reinforcement learning strategy, namely deep Q-Learning to schedule workloads in a heterogeneous fog computing environment~\cite{tang2018migration}. The state-of-the-art method, \textit{A3C}, schedules workloads using a policy gradient based reinforcement learning strategy which tries to optimize an actor-critic pair of agents~\cite{tuli2020dynamic}. This approach uses Residual Recurrent Neural Networks to predict the expected reward for each action \emph{i.e.}, scheduling decision and tries to optimize the cumulative reward signal.
However, such methods are still slow to adapt to real-world application scenarios as discussed in Section~\ref{sec:perf_eval}. This leads to higher wait times and subsequently high response times and SLO violations. Moreover, these methods do not scale well~\cite{nandi2001artificial}. 

Finally, coupled or symbiotic simulation and model based control have long been used in the modelling and optimization of distributed systems~\cite{tuli2019fogbus, bosmans2019testing, onggo2018symbiotic}. Many prior works have used hybrid simulation models to optimize decision making in dynamic systems. To achieve this, they monitor, analyze, plan and execute decisions using previous knowledge-base corpora (MAPE-k)~\cite{gill2019transformative}. However, such works use this to facilitate search methods and not to generate additional data to aid the decision making of an AI model. COSCO is developed to leverage a seamless interface between the orchestration framework and simulation engine to have a interactive dynamic between AI models to optimize QoS.

For our experiments, we use LR-MMT, MAD-MC, POND, GA, DQLCM and A3C as baselines to compare against the proposed GOBI and GOBI* models. We try to cover the best approach(es) in each category by careful literature review. The different categories have been selected for their complementary benefits. Heuristic based approaches are fast but not as accurate as other methods. Learning methods are adaptive but slow to adapt or output a scheduling decision. 

\section{Motivating Example}
\label{sec:motivation}

\begin{figure}[]
    \centering
    \includegraphics[width=0.85\columnwidth]{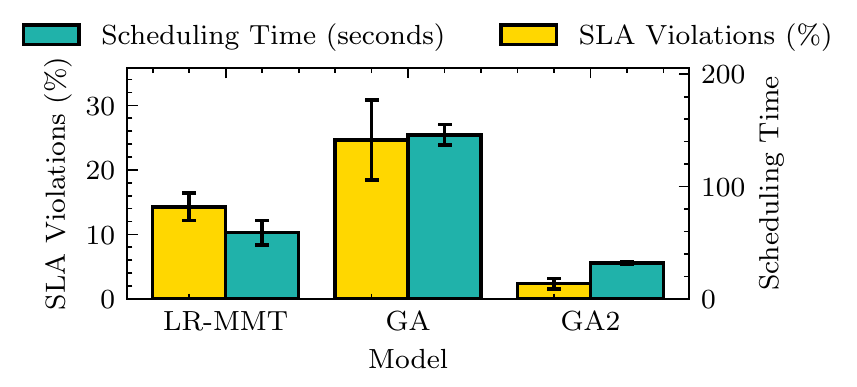}
    \caption{Motivating example}
    \label{fig:motivation}
\end{figure}

As demonstrated by many prior works~\cite{tuli2020dynamic, beloglazov2012optimal, han2018fog}, for fog datacenters, the scheduling time increases exponentially with the number of fog devices. This can have a significant impact on response time and, correspondingly, on Service Level Objectives (SLOs) corresponding to the task response time. To test this, we perform a series of experiments on the iFogSim simulator~\cite{gupta2017ifogsim}. We test the popular heuristic based approach LR-MMT~\cite{beloglazov2012optimal} and the evolutionary approach using genetic algorithm (GA)~\cite{han2018fog} with 50 simulated host machines (details in Section~\ref{sec:setup}). The former represents the broad range of heuristic based approaches that can quickly generate a scheduling decision but lack the ability to adapt. The latter have the ability to quickly adapt but have high scheduling overheads. As in~\cite{han2018fog}, we provide the GA the utilization metrics of all hosts and tasks and aim to minimize the average response time.

We assume a central scheduler which periodically allocates new tasks and migrates active tasks if required to reduce SLO violations. We assume that the container allocation time is negligible and the total response time of a fog task is the sum of the scheduling and execution times. We consider volatile workloads that are generated from traces of real-world applications running in a fog environment~\cite{tuli2020dynamic} (using the Bitbrain dataset, details in Section~\ref{sec:perf_eval}). For the approaches considered, the scheduling time can range from 100 to 150 seconds with SLO violation rate of up to $24\%$. However, if we could bring down the scheduling times to $10$ seconds with the same decisions as GA (referred to as GA2), the SLO violation rate goes down to $\sim \! \! 2\%$. This shows how scheduling time is a key factor that limits the extent to which SLO can be optimized (Figure~\ref{fig:motivation}), thus highlighting the potential for an approach which can respond quickly in volatile scenarios and has low scheduling overheads, leveraging the fast and scalable execution times of AI methods.



\section{System Model and Problem Formulation} 
\label{sec:system}


\subsection{System Model}
We consider a standard distributed and heterogeneous fog computing environment as shown in Figure~\ref{fig:system}. We assume a single fog datacenter with geographically distributed edge and cloud layers as computational nodes. We consider network latency effects for interactions between compute devices in different layers (edge and cloud) and between compute devices and fog broker. We ignore communication latencies among devices in the same layer.. 
Tasks are container instances with inputs being generated by sensors and other IoT devices and results received by actuators which physically manifest the intended actions. These devices constitute the IoT layer and send and receive all data from the fog gateway devices. All management of tasks and hosts is done by a Fog Broker in the management layer. This includes task scheduling, data management, resource monitoring and container orchestration. The work described in this paper is concerned with improving the scheduler in the Fog Broker by providing an interface between the container orchestration and resource monitoring services with a simulator. Our work proposes discrete-time controllers, that are commonly adopted in the literature~\cite{tuli2019fogbus, basu2019learn, tuli2020dynamic}. The communication between end-users and the fog broker is facilitated by gateway devices.
The compute nodes in the fog resource layer, which we refer to as ``hosts'', can have diverse computational capabilities. The hosts at the edge of the network are resource-constrained and have low communication latency with the broker and gateway devices. On the other hand, cloud hosts are several hops from the users and have much higher communication latency and are computationally more powerful. We assume that there are a fixed number of host machines in the fog resource layer and denote them as $H = \{h_0, h_1, \ldots, h_{N-1}\}$. We denote the collection of time-series utilization metrics, which include CPU, RAM, Disk and Network Bandwidth usage of host $h_i$ as $\mathcal{U}(h_i^t)$. The collection of maximum capacities of CPU, RAM, Disk and Network Bandwidth with the average communication latency of host $h_i$ is denoted as $\mathcal{C}(h_i)$. 

\begin{figure}[]
    \centering
    \includegraphics[width=0.95\columnwidth]{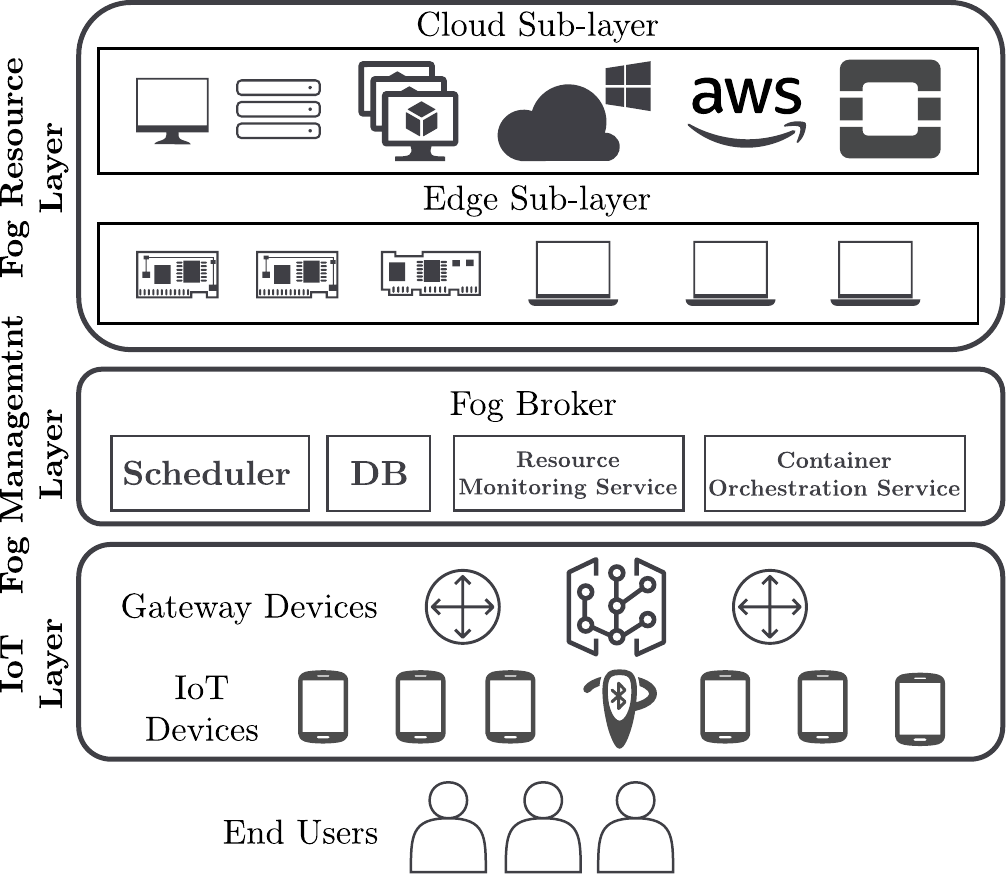}
    \caption{System model}
    \label{fig:system}
\end{figure}

\subsection{Workload Model}
We now present our workload model, as summarized in Figure~\ref{fig:workload}. We divide a bounded timeline into equal sized scheduling intervals with equal duration $\Delta$. The $t$-th interval is denoted by $I_t$ and starts at $s(I_t)$, so that $s(I_0) = 0$ and $s(I_t) = s(I_{t-1}) + \Delta\ \forall t>0$. In interval $I_{t-1}$, $N_t$ new tasks are created by the IoT devices. The gateway devices then send batches of these new tasks $N_t$ with their SLO requirements to the Fog Broker. The SLO violations are measured corresponding to the response time metrics of task executions as a fraction of them that exceeds the stipulated deadlines. The set of active tasks in the interval $I_t$ is denoted as $A_t = \{a_0^t, a_1^t, \ldots, a_{|A_t|}^t\}$ and consists of $|A_t|$ tasks. Here $a_j^t$ represents an identifier for the $j$-th task in $A_t$. The Fog Broker schedules the new tasks to the compute nodes and decides which of the active tasks need to be migrated. At the end of interval $I_{t-1}$, the set of completed tasks is denoted as $L_t$, hence the tasks that carry forward to the next interval is the set $A_{t-1} \setminus L_t$. Thus, the broker takes a decision $\mathcal{D}^t = \mathcal{D}(Y_t)$, where $Y_t$ is the union of the new ($N_t$), active ($A_{t-1} \setminus L_t$) and waiting tasks ($W_{t-1}$). This decision is a set of allocations and migrations in the form of ordered-pairs of tasks and hosts (assuming no locality constraints). At the first scheduling interval, as there are no active or waiting tasks, the model only takes allocation decision for new tasks ($N_0$). Here, $\mathcal{D}$ denotes the scheduler such that $\mathcal{D}: Y_t \mapsto Y_t \times H$. Only if the allocation is feasible, \emph{i.e.}, depending on whether the target host can accommodate the scheduled task or not, is the migration/allocation decision executed. Initially, the waiting queue and the set of active and leaving actions are empty. The executed decision is denoted by $\hat{\mathcal{D}}$ and satisfies the following properties
\begin{gather}
\label{eq:dhat}
    \hat{\mathcal{D}}(A_t) \subseteq \mathcal{D}(Y_t), \text{where}\\
    \forall\ (a_{j}^t, h_{i}) \in \hat{\mathcal{D}}(A_t),\ \mathcal{U}(a_j^t) + \mathcal{U}(h_i^t) \leq \mathcal{C}(h_i).
\end{gather}
We use the notation $a_j^t \in \hat{\mathcal{D}}(A_t)$ to mean that this task was allocated/migrated. The set of new and waiting tasks $\hat{N}_t \cup \hat{W}_{t-1} \subseteq \hat{\mathcal{D}}(A_t)$, where $\hat{N}_t \subseteq N_t$ and $\hat{W}_{t-1} \subseteq W_{t-1}$, denotes those tasks which could be allocated successfully. The remaining new tasks, \emph{i.e.} $N_t \setminus \hat{N}_t$ are added to get the new wait queue $W_t$. Similarly, active tasks are denoted as $\hat{A}_{t-1} \subseteq A_{t-1}$. Hence, for every interval $I_t$,
\begin{gather}
    \label{eq:active}
    A_t \gets \hat{N}_t \cup \hat{W}_{t-1} \cup A_{t-1} \setminus L_{t}\\
    W_t \gets (W_{t-1} \setminus \hat{W}_{t-1}) \cup (N_t \setminus \hat{N}_t), \text{where}\\
    \label{eq:interval}
    W_{-1} = A_{-1} = L_{-1} = \emptyset.
\end{gather}
We also denote the utilization metrics of an active task $a_j^t \in A_t$ as $\mathcal{U}(a_j^t)$ for the interval $I_t$. We denote a simulator as $\mathcal{S}: \hat{\mathcal{D}}(A_t) \times \{\mathcal{U}(a_j^t) | \forall a_j^t \in A_t\} \mapsto \{\mathcal{U}(h_i^t) | \forall h_i \in H\} \times P$, which takes scheduling decision and container utilization characteristics to give host characteristics and values of QoS parameters. Here $P_t \in P$ is a set of QoS parameters such as energy consumption, response times, SLO violations, etc. Similarly, execution on a physical fog framework is denoted by $\mathcal{F}$. Hence, execution of interval $I_t$ on a physical setup is denoted as
\begin{equation}
\label{eq:framework_step}
    (\{\mathcal{U}(h_i^t) | \forall h_i \in H\}, P_t) \equiv \mathcal{F}(\hat{\mathcal{D}}(A_t),  \{\mathcal{U}(a_j^t) | \forall a_j^t \in A_t\}).
\end{equation}

\begin{figure}[]
    \centering
    \includegraphics[width=0.95\columnwidth]{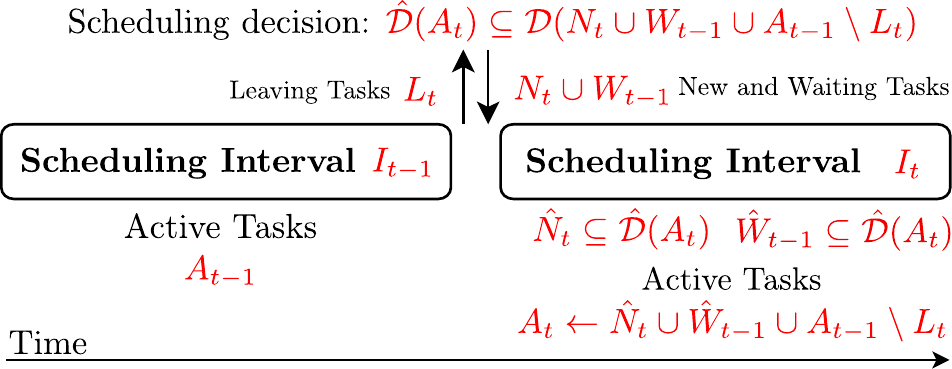}
    \caption{Dynamic Task Workload Model}
    \label{fig:workload}
\end{figure}

\begin{table}[]
    \centering
    \caption{Symbol Table}
    \resizebox{\columnwidth}{!}{
    \begin{tabular}{@{}ll@{}}
        \toprule
        Symbol & Meaning \\
        \midrule
        $I_t$ & $t^{th}$ scheduling interval \\ 
        $A_t$ & Active tasks in $I_t$ \\ 
        $W_{t-1}$ & Waiting tasks at the start of $I_t$ \\ 
        $L_t$ & Tasks leaving at the end of $I_t$ \\ 
        $N_t$ & New tasks received at the start of $I_t$ \\ 
        $H$ & Set of hosts in the Resource Layer \\ 
        $h_i$ & $i^{th}$ host in an enumeration of $H$ \\ 
        $a_j^t$ & $j^{th}$ task in an enumeration of $A_t$ \\ 
        $\mathcal{U}(h_i^t)$ & Utilization metrics of host $h_i$ in $I_t$\\ 
        $\mathcal{U}(a_j^t)$ & Utilization metrics of task $a_j^t$ in $I_t$\\ 
        $Y_t = N_t \cup W_{t-1} \cup A_{t-1} \setminus L_t$ & Scheduler input at the start of $I_t$\\ 
        $\mathcal{D}(Y_t)$ or $\mathcal{D}$ & Scheduling decision at start of $I_t$\\ 
        $\hat{\mathcal{D}}$ & Feasible sub-set of scheduling decision $\mathcal{D}$\\  \xrowht{4pt}
        $\bar{\mathcal{D}}$ & Scheduling decision of GOBI in GOBI* loop\\ 
        $\mathcal{O}(P_t)$ & Objective value at the end of $I_t$\\ 
        $\mathcal{S}$ & Execution of an interval on simulator \\ 
        $\mathcal{F}$ & Execution of an interval on physical setup\\ \bottomrule
    \end{tabular}}
    \label{tab:symbols}
\end{table}

\subsection{Problem Formulation}

After execution of tasks $A_t$ in interval $I_t$, the objective value to minimize is denoted as $\mathcal{O}(P_t)$. $\mathcal{O}(P_t)$ could be a scalar value which combines multiple QoS parameters. To find the optimum schedule, we need to minimize the objective function $\mathcal{O}(P_t)$ over the entire duration of execution. Thus, we need to find the appropriate and feasible decision function $\mathcal{D}$ such that $\sum_t \mathcal{O}(P_t)$ is minimized. This is subject to the constraints that at each scheduling interval the new tasks $N_t$, waiting tasks $W_{t-1}$ and active tasks from the previous interval $A_{t-1} \setminus L_t$ are allocated using this decision function. The problem can then be concisely formulated as:

\begin{equation}
\label{eq:problem}
\begin{aligned}
& \underset{\mathcal{D}}{\text{minimize}}
& & \sum_{t=0}^T \mathcal{O}(P_t) \\
& \text{subject to}
& & \forall\ t, Eqs. \eqref{eq:dhat}-\eqref{eq:framework_step}.
\end{aligned}
\end{equation}


\section{The GOBI Scheduler}
\label{sec:gobi}

As discussed in Section~\ref{sec:introduction}, we now present the approach based on Gradient Based Optimization using back-propagation to Input (GOBI). These constitute the \textit{Scheduler} module of the Fog Broker described in Figure~\ref{fig:system}. Here, we consider that we optimize an objective function $\mathcal{O}(P_t)$ at every interval $I_t$ via taking an optimal action in the form of scheduling decision $\mathcal{D}$. This $\mathcal{D}$ is then used by the framework/simulator to execute tasks as described in optimization program~\eqref{eq:problem}. 

\subsection{Objective Function}
We will now present how our optimization scheme can be used for taking appropriate scheduling decisions in a fog environment. To optimize the QoS parameters, we consider an objective function that focuses on energy consumption and response time, two of the most crucial metrics for fog environments~\cite{mahmoud2018towards, mutlag2019enabling}. For interval $I_t$,

\begin{equation}
    \label{eq:objective_function}
    \mathcal{O}(P_t) = \alpha \cdot AEC_t + \beta \cdot ART_t.
\end{equation}
Here $AEC$ and $ART$ ($\in P_t$) are defined as follows.
\begin{enumerate}[leftmargin=*]
    \item \textit{Average Energy Consumption} (AEC) is defined for any interval as the energy consumption of the infrastructure (which includes all edge and cloud hosts) normalized by the maximum power of the hosts, \emph{i.e.},
    \begin{equation}
        \begin{aligned}
        \label{eq:aec}
        AEC_t = \frac{\sum_{h_i \in H} \int_{t=s(I_t)}^{s(I_{t+1})} Power_{h_i}(t) dt}{|A_t| \sum_{h_i\in H}  Power^{max}_{h_i} \times (t_{i+1} - t_i)},
        \end{aligned}
    \end{equation}
    where $Power_{h_i}(t)$ is the power function of host $h_i$ at instant $t$, and $Power_{h_1}^{max}$ is maximum possible power of $h_i$.
    \item \textit{Average Response Time} (ART)\footnote{Both AEC and ART are unit-less metrics and lie between $0$ and $1$} is defined for an interval $I_t$ as the average response time for all leaving tasks ($L_t$), normalized by maximum response time until the current interval, as shown below
    \begin{equation}
        \begin{aligned}
        \label{eq:art}
        ART_t = \frac{\sum_{l_j^t \in L_t} Response\ Time(l_j^t)}{|L_t| \max_{s\leq t} \max_{l_j^{s} \in L_{s}} Response\ Time(l_j^{s})}.
        \end{aligned}
    \end{equation}
\end{enumerate}
\subsection{Input Parameters}

In concrete implementations, we can assume a finite maximum number of active tasks, with the upper bound denoted as $M$. Thus at any interval, $|A_t| \leq M$. Moreover, we consider the utilization metrics of Instructions per second (IPS), RAM, Disk and Bandwidth consumption which form a feature vector of size $F$. Then, we express task utilization metrics $\{\mathcal{U}(a_j^{t-1}) | \forall j, a_j^{t-1} \in A_{t-1}\}$ as an $M\times F$ matrix denoted as $\phi(A_{t-1})$, where the first $A_{t-1}$ rows are the feature vectors of tasks in $A_{t-1}$ in the order of increasing creation intervals (breaking ties uniformly at random), and the rest of the rows are ${0}$. We also form feature vectors of hosts with IPS, RAM, Disk and Bandwidth utilization with capacities and communication latency of each host, each of size $F'$. Thus, at each interval $I_t$ for $|H|$ hosts, we form a $|H| \times F'$ matrix $\phi(H_{t-1})$ using host utilization metrics of interval $I_{t-1}$. 

Finally, as we have $N_t, W_{t-1}, A_{t-1}$ and $L_t$ at the start of $I_t$, the decision matrix $\phi(\mathcal{D})$ is an $M \times |H|$ matrix with the first $|N_t \cup W_{t-1} \cup A_{t-1} \cup L_t|$ rows being one-hot vector of allocation of tasks to one of the $H$ hosts. The remaining rows being ${0}$.

\subsection{Model Training}

We now describe how a neural model can be trained to approximate $\mathcal{O}(P_t)$ using the input parameters $[\phi(A_{t-1}), \phi(H_{t-1}), \phi(\mathcal{D})]$. Consider a continuous function $f(x;\theta)$ as a neural approximator of $\mathcal{O}(P_t)$ with the $\theta$ vector denoting the underpinning neural network parameters and $x$ as a continuous or discrete variable with a bounded domain. Here, $x$ is the collection of utilization metrics of tasks and hosts with a scheduling decision. The parameters $\theta$ are learnt using the dataset $\Lambda = \{ [\phi(A_{t-1}), \phi(H_{t-1}), \phi(\mathcal{D})], \mathcal{O}(P_t) \}_t$ such that a given loss function $\mathcal{L}$ is minimized for this dataset. We form this dataset by running a random scheduler and saving the traces generated (more details about the setup in Section~\ref{sec:perf_eval}). The loss $\mathcal{L}$ quantifies the dissimilarity between the predicted output and the ground truth. We use Mean Square Error (MSE) as the loss function as done in prior work~\cite{tuli2020dynamic}. Hence,
\[ \mathcal{L}(f(x;\theta), y) = \tfrac{1}{T} \textstyle \sum_{t=0}^T (y - f(x;\theta))^2, \text{ where} (x,y) \in \Lambda.\]

Thus, for datapoints $(x,y) \in \Lambda$, where $y$ is the value of $\mathcal{O}(P_t)$ for the corresponding $x$, we have $\theta = \argmin_{\hat{\theta}}{\sum_{(x,y)\in\Lambda} [ \mathcal{L}(f(x;\hat{\theta}), y) ]}$. To do this, we calculate the gradient of the loss function with respect to $\theta$ as $\nabla_\theta \mathcal{L}$ and use back-propagation to learn the network parameters. By the universal approximation theorem~\cite{scarselli1998universal}, we can assume that we already have parameters $\theta$ after network training, such that this neural network approximates a generic function to an arbitrary degree of precision (considering a sufficiently large parameter set $\theta$). A key part of this ability of neural networks rests on randomly initializing a parameterized function and changing the large number of parameters based on the loss function.

Now, with this, we need to find the appropriate decision matrix $\phi(\mathcal{D})$, such that $f(x;\theta)$ is minimized. We formulate the resulting optimization problem as
\begin{equation}
\begin{aligned}
& \underset{\phi(\mathcal{D})}{\text{minimize}}
& & f(x;\theta) \\
& \text{subject to}
& & \text{each element of }\phi(\mathcal{D}) \text{ is bounded}\\
&&& \forall\ t, Eqs. \eqref{eq:dhat}-\eqref{eq:framework_step}.
\end{aligned}
\end{equation}

This is a reformulation of the program~\eqref{eq:problem} using the neural approximator of $\mathcal{O}(P_t)$ as the objective function. Note that, because there is a bounded space of inputs, there must be an optimal solution, \emph{i.e.}, $\exists\ \hat{\phi}(\mathcal{D})$ such that 

\vspace{4pt}
\begin{gather*}
f([\phi(A_{t-1}), \phi(H_{t-1}), \hat{\phi}(\mathcal{D})];\theta) \leq \\
\null \qquad \qquad f([\phi(A_{t-1}), \phi(H_{t-1}), \phi(\mathcal{D})];\theta),\ \forall \text{ feasible } \phi(\mathcal{D}).
\end{gather*}
\vspace{4pt}

We solve this optimization problem using gradient-based methods, for which we need $\nabla_x f(x;\theta)$. This is uncommon for neural network training since prior work normally modifies the weights using input and ground truth values. Thus the gradients can be calculated without having to update the weights. Now, for a typical feed-forward artificial neural network, $f(x;\theta)$ is a composition of linear layers and non-linear activation functions. Just like the back-propagation approach, we can find gradients with respect to input $x$. For simplicity, we consider non-affine functions like $\textsf{tanh()}$ and $\textsf{softplus()}$ as activations in our models, with no effect to the generality of the neural approximator. For instance, gradients for a single linear layer with $\textsf{tanh()}$ non-linearity is given in Proposition~\ref{backpropinput}\footnote{Derivations for other activation functions with a single linear layer are given in Appendix~B.}.

\begin{theorem}[Finding gradients with respect to input]
\label{backpropinput}
For a linear layer with \textsf{tanh}() non-linearity that is defined as follows:
\[ f(x;W,b) = \textsf{tanh}(W\cdot x + b), \]
the derivative of this layer with respect to $x$ is given as:
\[ \nabla_x f(x;W,b) = W^T \times (1 - \textsf{tanh}^2(W\cdot x + b)).\]
\end{theorem}
\begin{proof}
It is well known that $\nabla_x \textsf{tanh}(x) = 1- \textsf{tanh}^2(x)$, therefore,
\begin{align*}
    \nabla_x f(x;W,b) &= \nabla_x \textsf{tanh}(W\cdot x + b)\\
    &= \nabla_{W\cdot x + b} \textsf{tanh}(W\cdot x + b) \times \nabla_x (W\cdot x + b)\\
    &= (1 - \textsf{tanh}^2(W\cdot x + b)) \times \nabla_x (W\cdot x + b)\\
    &= W^T \times (1 - \textsf{tanh}^2(W\cdot x + b)) \qedhere
\end{align*}
\end{proof}

Using the chain-rule, we can find the derivative of an arbitrary neural approximator $f$ as $f(x;\theta)$ is a composition of multiple such layers. Once we have the gradients, we can initialize an arbitrary input $x$ and use gradient descent to minimize $f(x;\theta)$ as
\begin{equation}
    \label{eq:update}
    x_{n+1} \gets x_n - \gamma \cdot \nabla_x f(x_{n};\theta),
\end{equation}
where $\gamma$ is the learning rate and $n$ is the iteration count. We apply Eq.~\eqref{eq:update} until convergence. The calculation of $\nabla_x f(x_{n};\theta)$ in the above equation is carried out in a similar fashion as back-propagation for model parameters.
We can use momentum based methods like Adam~\cite{kingma2014adam} or annealed schedulers like cosine annealing with warm restarts~\cite{loshchilov2016sgdr} to prevent the optimization loop getting stuck in local optima.
As we show in our experiments, the combination of gradient calculation with respect to inputs and advances like momentum and annealing allow \emph{faster} optimization \emph{avoiding} local minima, giving a more robust and efficient approach.

\begin{algorithm}[t]
    \begin{algorithmic}[1]
    \Require
    \Statex Pre-trained function approximator $f(x;\theta)$
    \Statex Dataset used for training $\Lambda$; Convergence threshold $\epsilon$
    \Statex Iteration limit $\sigma$; Learning rate $\gamma$
    \Statex Initial random decision $\mathcal{D}$
    \Procedure{Minimize}{$\mathcal{D}$, $f$, $z$}
        \State Initialize decision matrix $\phi(\mathcal{D})$; $i = 0$
        \State \textbf{do}
        \State \hspace{\algorithmicindent} $x \gets [z, \phi(\mathcal{D})]$ \Comment{Concatenation}
        \State \hspace{\algorithmicindent} $\delta \gets \nabla_{\phi(\mathcal{D})}f(x; \theta)$ \Comment{Partial gradient}
        \State \hspace{\algorithmicindent} $\phi(\mathcal{D}) \gets \phi(\mathcal{D}) - \gamma \cdot \delta$ \Comment{Decision update}
        \State \hspace{\algorithmicindent} $i \gets i + 1$
        \State \textbf{while} $|\delta| > \epsilon$ and $i \leq \sigma$
        \State Convert matrix $\phi(\mathcal{D})$ to scheduling decision $\mathcal{D}^*$
        \State \textbf{return} $\mathcal{D}^*$
    \EndProcedure
    \Procedure{GOBI}{scheduling interval $I_t$}
        \State \textbf{if} (t == 0)
        \State \hspace{\algorithmicindent} Initialize random decision $\mathcal{D}$
        \State \textbf{else}
        \State \hspace{\algorithmicindent} $\mathcal{D} \gets \mathcal{D}^*$ \Comment{Output for the previous interval}
        \State Get $\phi(A_{t-1}), \phi(H_{t-1})$
        \State $\mathcal{D}^* \gets \textsc{Minimize}(\mathcal{D}, f, [\phi(A_{t-1}), \phi(H_{t-1})])$
        \State Fine-tune $f$ with loss $=$
        \State \ \ \ $MSE(\mathcal{O}(P_{t-1}), f([\phi(A_{t-2}), \phi(H_{t-2}), \phi(\mathcal{D}^{t-1})] ; \theta))$
        \State \textbf{return} $\mathcal{D}^*$
    \EndProcedure
    \end{algorithmic}
\caption{The GOBI scheduler}
\label{alg:gobi}
\end{algorithm}

\subsection{Scheduling}

After training the model $f$, at start of each interval $I_t$, we optimize $\phi(\mathcal{D})$ by the following rule ($\gamma$ is the learning rate) until the absolute value of the gradient is more than the convergence threshold $\epsilon$, \emph{i.e.},

\begin{equation}
\label{eq:gobi_teration}
    \phi(\mathcal{D})_{n+1} \gets \phi(\mathcal{D})_{n} - \gamma \cdot \nabla_{\phi(\mathcal{D})} f(x_n; \theta).
\end{equation}

In dynamic scenarios, as the approximated function for the objective $\mathcal{O}(P_{t})$ changes with time we continuously fine-tune the neural approximator as per the loss function $MSE(\mathcal{O}(P_{t-1}), f([\phi(A_{t-2}), \phi(H_{t-2}), \phi(\mathcal{D}^{t-1})] ; \theta))$ (line 20 in Algorithm~\ref{alg:gobi}). Equation~\ref{eq:gobi_teration} then leads to a decision matrix with a lower objective value (line 6 in Algorithm~\ref{alg:gobi}). Hence, \textit{gradient based iteration with continuous fine-tuning of the neural approximator allows GOBI to adapt quickly in dynamic scenarios.} When the above equation converges to $\phi(\mathcal{D}^*)$, the rows represent the likelihood of allocation to each host. We take the $\argmax$ over each row to determine the host to which each task should be allocated as $\mathcal{D}^*$. This then becomes the final scheduling decision of the GOBI approach. The returned scheduling decision is then run on a simulated or physical platform (as in Eq.~\eqref{eq:framework_step}) and tasks sets are updated as per Eqs.~\eqref{eq:active}-\eqref{eq:interval}. To find the executed decision $\hat{\mathcal{D}}$, we sort all tasks in $\mathcal{D}^*$ in descending order of their wait times (breaking ties uniformly at random) and then try to allocate them to the intended host in $\mathcal{D}^*$. If the allocation is not possible due to utilization constraints, we ignore such migrations and do not add them to $\hat{\mathcal{D}}$. At run-time, we fine tune the trained neural model $f$, by newly generated data for the model to adapt to new settings.


\section{The COSCO Framework}
\label{sec:cosco}


Motivated from prior coupled-simulators~\cite{bosmans2019testing}, we now present the COSCO framework as introduced in Section~\ref{sec:introduction}. A basic architecture is shown in Figure~\ref{fig:arch}. Implementation details of the COSCO framework are given in Appendix~A. The hosts might correspond to simulated hosts or physical compute nodes with Instructions per second (IPS), RAM, Disk and Bandwidth capacities with power consumption models with CPU utilization. The workloads are either time-series models of IPS, RAM, Disk, Bandwidth requirements in the case of simulation or are actual application programs in the case of physical experiments. At the beginning of each scheduling interval, $N_t$ workloads are created. The framework also maintains a waiting queue for workloads that cannot be allocated in the previous interval ($W_{t-1}$).

\begin{figure}[]
    \centering
    \includegraphics[width=0.7\columnwidth]{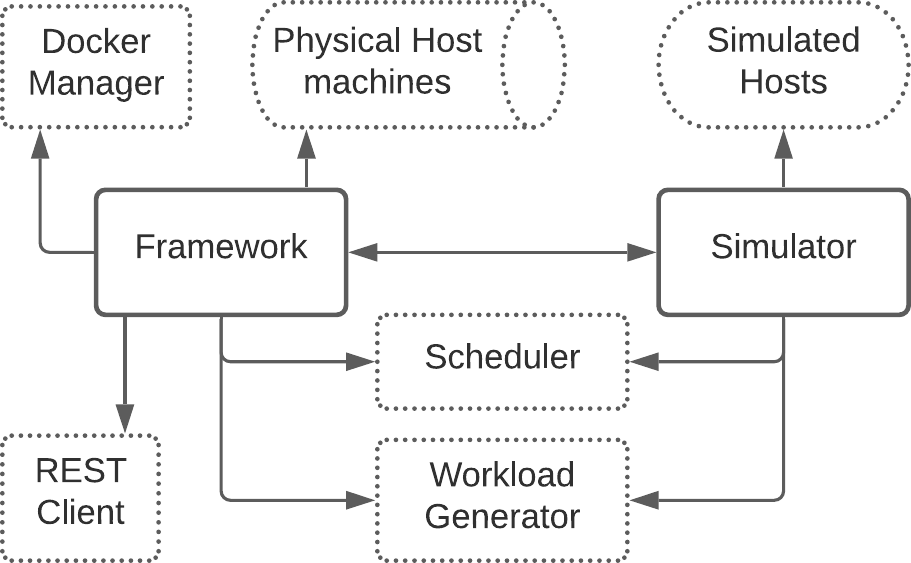}
    \caption{COSCO architecture}
    \label{fig:arch}
\end{figure}


\subsection{Simulator}
\label{sec:simulator}
The \textit{Simulator} uses a \textit{Scheduler} (GOBI for instance), which uses the utilization metrics and QoS parameters to take an allocation decision for new workloads or migration decisions for active workloads. At the start of a scheduling interval $I_t$, the simulator destroys workloads completed in $I_{t-1}$, denoted as $L_t$. It also gets workloads $N_t \cup W_{t-1}$. The \textit{Scheduler} decides which host to allocate these tasks to and  whether to migrate active tasks from $A_{t-1} \setminus L_t$. The \textit{Simulator} object considers all $N_t \cup W_{t-1}$ tasks in the order of the interval at which they were created. If the target host is unable to accommodate a task, it is added to the wait queue. This is done for each utilization metric of active task $a_j^t$ in $\mathcal{U}(a^t_j)$. As for new tasks, the utilization metrics are unknown, $u\, =\, 0\, \forall u \in \mathcal{U}(n^t_j)\, \forall\, n^t_j \in N_t$. To prevent overflow of host resources, we compare for each new task $n_j^t$ the sum of the maximum possible utilization metrics of the corresponding application and target host against the host's maximum capacity (more details in Appendix A). Hence, at each interval, Eq.~\eqref{eq:interval} holds. Instead of running the tasks on physical machines, we use trace driven discrete-event simulation.

\subsection{Framework}
\label{sec:framework}
The \textit{Framework} instantiates the tasks as Docker containers. Docker is a container management platform as a service used to build, ship and run containers on physical or virtual environments. Moreover, the hosts are physical computational nodes in the same Virtual Local Area Network (VLAN) as the server machine on which the \textit{Scheduler} program runs. The server communicates with each host via HTTP REST APIs~\cite{tuli2019fogbus}. At every host machine, a \texttt{DockerClient} service runs with a \texttt{Flask} HTTP web-service~\cite{grinberg2018flask}. Flask is a web application REST client which allows interfacing between hosts and the server. To migrate a container, we use the Checkpoint/Restore In Userspace (CRIU)~\cite{venkatesh2019fast} tool. For each migration decision $(a^t_j, h_i) \in \hat{\mathcal{D}}(A_t)$, a checkpoint of $a^t_j$ is created which is equivalent to a snapshot of the container. Then, this snapshot is transferred (migrated) to the target host $h_i$, where the container is resumed (restored). The allocation and migrations are done in a non-blocking fashion, \emph{i.e.} the execution of active containers is not halted when some tasks are being migrated.

\subsection{Model Interface}
\label{sec:interface}

We now discuss the novelty of the COSCO framework.
As all scheduling algorithms and simulations are run on a central fog broker~\cite{tuli2019fogbus}, it can be easily run in fog environments which follow a master-slave topology. As discussed in Section~\ref{sec:system}, a simulator can execute tasks on simulated host machines to return QoS parameters $P$. To execute an interval $I_t$, the simulator needs utilization metrics $\mathcal{U}(a^t_j), \forall a^t_j \in A_t$, and a scheduling decision $\bar{\mathcal{D}}^t$ to be executed on the simulator. Thus, at the beginning of the interval $I_t$, we have $L_t, N_t, W_{t-1}$ and $A_{t-1}$. After checking for the possibility of allocation/migration, we get $\hat{\bar{\mathcal{D}}}(A_t)$, $\hat{N}_t$ and $\hat{W}_{t-1}$. Now, using Eq.~\eqref{eq:active}, we find $A_t$. Using given utilization metrics $\mathcal{U}(a^t_j)$ and $\hat{\bar{\mathcal{D}}}(A_t)$, we execute $A_t$ tasks on the simulator to get $\mathcal{S}(\hat{\bar{\mathcal{D}}}(A_t), \{\mathcal{U}(a_j^t) | \forall a_j^t \in A_t\}) = \{\mathcal{U}(h_i^t) | \forall h_i \in H\}, \bar{P}_t$. This means that at the beginning of interval $I_t$, the COSCO framework allows simulation of the next scheduling interval (with an action of interest $\bar{\mathcal{D}}$ and predicted utilization metrics) to predict the values of the QoS parameters in the next interval $\bar{P}_t$. This single step look-ahead simulation allows us to take better scheduling decisions in the GOBI* algorithm, as shown in the next section. 

\begin{algorithm}[t]
    \begin{algorithmic}[1]
    \Require
    \Statex Pre-trained function approximator $f^*(x;\theta^*)$
    \Statex Pre-trained LSTM model $LSTM(\{ \mathcal{U}(a^{t'}_j) \forall t'<t \})$
    \Statex Dataset used for training $\Lambda^*$; Convergence threshold $\epsilon$
    \Statex Learning rate $\gamma$; Initial random decision $\mathcal{D}$
    \Procedure{$\text{GOBI}^*$}{scheduling interval $I_t$}
        \State \textbf{if} (t == 0)
        \State \hspace{\algorithmicindent} Initialize random decision $\mathcal{D}$
        \State \textbf{else}
        \State \hspace{\algorithmicindent} $\mathcal{D} \gets \mathcal{D}^*$ \Comment{Output for the previous interval}
        \State Get $\phi(A_{t-1}), \phi(H_{t-1})$
        \State $\bar{\mathcal{U}}(a^t_j) \forall a^t_j \in A_t \gets LSTM(\{ \mathcal{U}(a^{t'}_j) \forall t'<t \})$
        \State $\bar{\mathcal{D}} \gets GOBI(I_t)$
        \State $\{\mathcal{U}(h_i^t) | \forall h_i \in H\}, \bar{AEC}_t, \bar{ART}_t \gets$
        \State \qquad \qquad \qquad $\mathcal{S}(\hat{\bar{\mathcal{D}}}(A_t),  \{\bar{\mathcal{U}}(a^t_j) \forall a^t_j \in A_t\})$
        \State $\mathcal{O}(\bar{P}_t) \gets \alpha \cdot \bar{AEC}_t + \beta \cdot \bar{ART}_t$
        \State $\mathcal{D}^* \gets \textsc{Minimize}(\mathcal{D}, f, [\phi(A_{t-1}), \phi(H_{t-1}), \mathcal{O}(\bar{P}_t)])$
        \State Fine-tune $f^*$ with loss $=$
        \State \ \ \ $MSE(\mathcal{O}(P_{t-1}), f([\phi(A_{t-2}), \phi(H_{t-2}), \phi(\mathcal{D}^{t-1})] ; \theta)) + \mathbbm{1}(\mathcal{O}(\bar{P}_{t-1}) < \mathcal{O}(P_{t-1})) \times MSE(\bar{\mathcal{D}}^{t-1}, \mathcal{D}^{t-1})$
        \State \textbf{return} $\mathcal{D}^*$
    \EndProcedure
    \end{algorithmic}
\caption{The GOBI* scheduler}
\label{alg:gobi2}
\end{algorithm}

\section{Extending GOBI to GOBI*}
\label{sec:gobi2}
Now that we have a scheduling approach based on back-propagation of gradients, we extend it to incorporate simulated results as described in Section~\ref{sec:introduction}. To do this, we use the following components.

\begin{enumerate}[leftmargin=*]
    \item \textit{GOBI Scheduler}: We assume that we have the GOBI scheduler which can give us a preferred action of interest $\bar{\mathcal{D}}$. Thus, at the start of a scheduling interval $I_t$, we get $\bar{\mathcal{D}} = GOBI(I_t)$ (line 8 in Algorithm~\ref{alg:gobi2}). $\mathcal{D}$ now denotes the decision of GOBI*.
    \item \textit{Utilization prediction models}: We train a utilization metric prediction model, such that using previous utilization metrics of the tasks, we get a predicted utilization metric set for the next interval. As is common in prior work~\cite{gupta2017resource}, we use a Long-Short-Term-Memory (LSTM) neural network for this and train it using the same $\Lambda$ dataset that we used for training the GOBI neural approximator. Thus, at the start of interval $I_t$, using $\bar{\mathcal{D}}$ from GOBI and checking allocation possibility, we get $\hat{N_t}$ and $\hat{W}_{t-1}$. Using (\ref{eq:active}), we get $A_t$. Then, we predict $\bar{\mathcal{U}}(a^t_j), \forall a^t_j \in A_t$, using the LSTM model (line 7 in Algorithm~\ref{alg:gobi2}). Hence, we get the $\phi(A_{t-1}), \phi(H_{t-1})$ matrices.
    \item \textit{Simulator}: Now that we have an action $\bar{\mathcal{D}}$ and predicted utilization metrics $\bar{\mathcal{U}}(a^t_j), \forall a^t_j \in A_t$, we can use the simulator to predict the QoS parameters at the end of $I_t$ as described in Section~\ref{sec:interface} (lines 9 and 10 in Algorithm~\ref{alg:gobi2}). 
\end{enumerate}

\begin{figure}[]
    \centering
    \includegraphics[width=0.95\columnwidth]{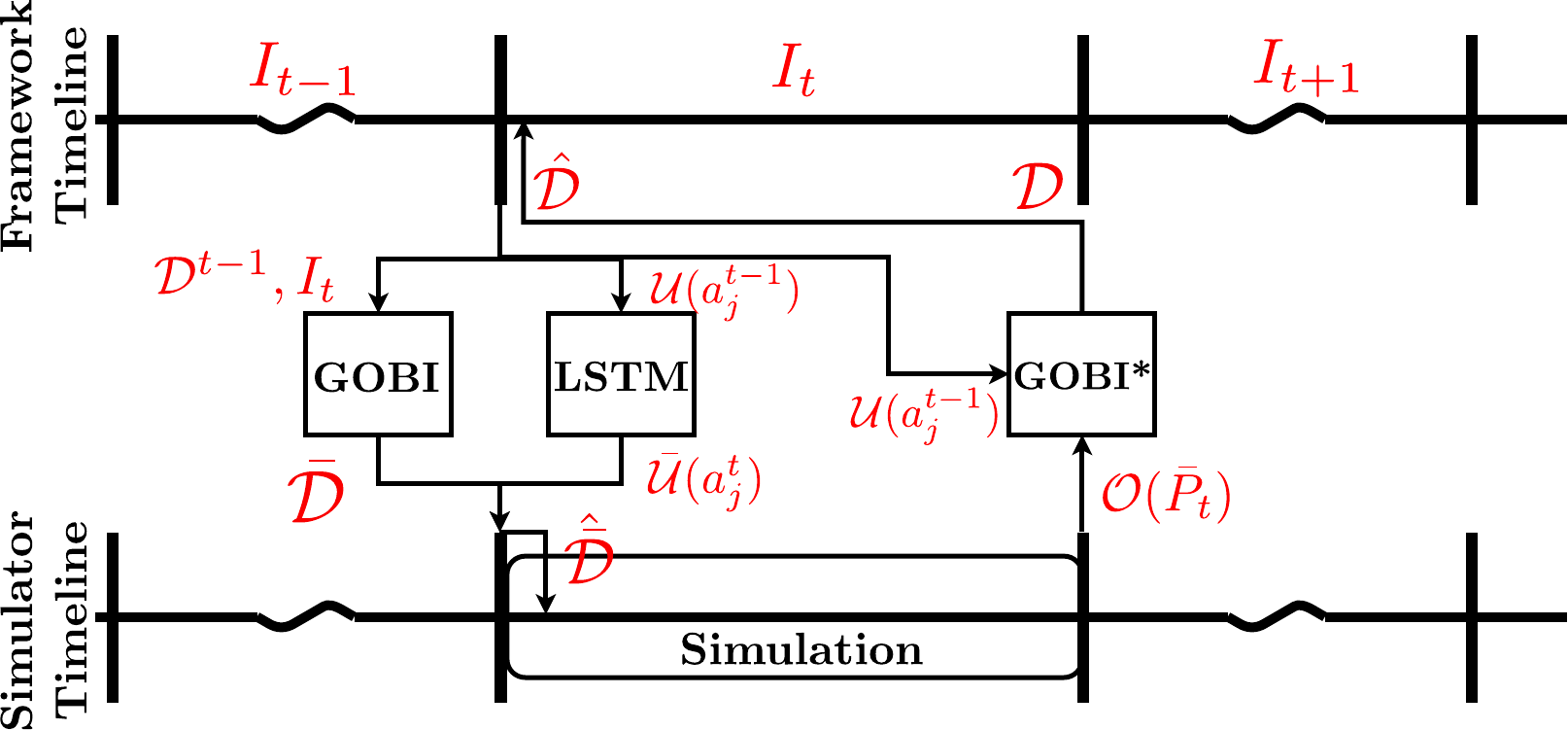}
    \caption{Iteration of GOBI* approach at the start of interval $I_t$}
    \label{fig:gobi2}
\end{figure}

Since at the start of the interval $I_t$, we do not have utilization models of all containers and subsequently hosts for $I_t$, our GOBI model is forced to predict QoS metrics at the end of $I_t$ using only utilization metrics of the previous interval, \emph{i.e.}, $I_{t-1}$. This can make GOBI's predictions inaccurate in some cases. However, if we can predict with reasonable accuracy the utilization metrics in $I_t$ using the utilization values of the previous intervals, we can simulate and get an improved estimate of the QoS parameters during the next period $I_t$. Adding these as inputs to another neural approximator for $\mathcal{O}(P_t)$ denoted as $f^*(x, \theta^*)$, we now have $x = [\phi(A_{t-1}), \phi(H_{t-1}), \mathcal{O}(\bar{P}_t), \phi(\mathcal{D})]$. 

Again, using a random scheduler and pre-trained GOBI model, we generate a dataset $\Lambda^*$ to train this new neural approximator $f^*$. At execution time, we freeze the neural model $f$ of GOBI and fine tune $f^*$ with the loss function. This loss function is the MSE of objective values. We add the MSE between the predicted action $\mathcal{D}$ by GOBI* and $\bar{\mathcal{D}}$ by GOBI in case the estimated objective value of GOBI is lower than that of GOBI*. Thus,
\begin{multline}
    \label{eq:loss_gobi2}
    \mathcal{L} = MSE(\mathcal{O}(P_{t-1}), f^*([\phi(A_{t-2}), \phi(H_{t-2}), \phi(\mathcal{D}^{t-1})] ; \theta^*)) \\
     + \mathbbm{1}(\mathcal{O}(\bar{P}_{t-1}) < \mathcal{O}(P_{t-1})) \times MSE(\bar{\mathcal{D}}^{t-1}, \mathcal{D}^{t-1}).
\end{multline}
Using this trained model, we can now provide schedules as shown in Algorithm~\ref{alg:gobi2} and Figure~\ref{fig:gobi2}. At interval $I_t$, the previous output of GOBI*, \emph{i.e.}, $\mathcal{D}^{t-1}$ is given to the GOBI model as the initial decision. The GOBI model uses utilization metrics to output the action $\hat{\mathcal{D}}$. Using a single-step simulation, we obtain an estimate of the objective function for this decision. Then, GOBI* uses this estimate to predict the next action $D$. Based on whether GOBI*'s decision is better than GOBI or not, GOBI* is driven to the decision which has lower objective value. \textit{This interactive dynamic should allow GOBI* to make a more informed prediction of QoS metrics for $I_t$ and hence perform better than GOBI.} 

\subsection{Convergence of GOBI and GOBI*}

We outline two scenarios of an ideal and a more realistic case to describe the interactive dynamic between GOBI and GOBI* and how the latter optimizes the objective scores.

\textbf{Ideal Case:} Consider the case when the neural approximator of GOBI, $f$ can \textit{perfectly} predict the objective value of the next interval $I_t$ as $\mathcal{O}(P_t)$. In this case, the predicted action $\bar{\mathcal{D}}$ is optimal. In such a setting, we expect GOBI* to converge such that it predicts the same action as $\bar{\mathcal{D}}$. Assuming an ideal simulator which \textit{exactly} mimics the real-world, $\mathcal{O}(\bar{P}_t) = \mathcal{O}(P_t)$. Thus, any action other than $\bar{\mathcal{D}}$ is sub-optimal. Hence, to minimize the loss metric, GOBI* would directly forward $\bar{\mathcal{D}}$ and ignore all other information and converge to GOBI (which is optimal).

\textbf{Real Case:} While the limited explainability of neural networks prevent us from giving formal convergence properties, we supply of number of technical observations that justify the convergence we have seen in practice when applying the methods to real systems. Thus, considering the real-case, time-bound training only gives a sub-optimal approximator $f$ for GOBI. Hence, in general the predicted decision ${\mathcal{D}}$ is not always optimal. Assuming GOBI* predicts a decision other than GOBI, GOBI* will sometimes get a sub-optimal action ${\mathcal{D}}$ such that $\mathcal{O}(\bar{P}_t) < \mathcal{O}(P_t)$ and other times $\mathcal{O}(\bar{P}_t) > \mathcal{O}(P_t)$. The former is the case when the action predicted by GOBI, $\bar{\mathcal{D}}$, is better and the latter is the case when the GOBI* action, ${\mathcal{D}}$, has a lower objective score.  (1) In the case of ${\mathcal{D}}$ being better than $\bar{\mathcal{D}}$ (having lower objective score), in an attempt to minimize the first part of the loss, 
\begin{equation}
    \label{eq:loss_part_1}
    \mathcal{L}_1 = MSE(\mathcal{O}(P_{t-1}), f^*([\phi(A_{t-2}), \phi(H_{t-2}), \phi(\mathcal{D}^{t-1})] ; \theta^*)),
\end{equation}
GOBI* would converge to avoid $\bar{\mathcal{D}}$ (considering the ideal simulator that returns exact answers). Furthermore, the predicted action $\mathcal{D}$ has a lower objective score, hence the optimization loop would tend to converge to this or better scheduling decisions. (2) In the case of $\bar{\mathcal{D}}$ being better than ${\mathcal{D}}$, which means that $\mathcal{O}(\bar{P}_{t-1}) < \mathcal{O}(P_{t-1})$, GOBI* is given an incentive to predict $\bar{\mathcal{D}}$. This is because the latter part of the loss makes the model converge such that it minimizes
\begin{equation}
    \label{eq:loss_part_2}
    \mathcal{L}_2 = MSE(\bar{\mathcal{D}}^{t-1}, \mathcal{D}^{t-1}).
\end{equation}
This, with the former part of the loss, creates an interactive dynamic such that the GOBI* model can find scheduling decisions even better than the one predicted by the GOBI model. However, in cases when GOBI*'s prediction is worse, the GOBI's prediction aids the model to quickly converge to a better decision (that of GOBI's). Thus, in both cases, each model update based on the loss in Eq.~\eqref{eq:loss_gobi2}, yields a scheduling decision with objective score same or lower than that given by GOBI. Assuming a sufficiently long training time, GOBI* would eventually converge to predict $\mathcal{D}$ such that the objective value when the decision is $\mathcal{D}$ is never greater than when it is $\bar{\mathcal{D}}$. We observe in our experiments that in the real-case GOBI* always performs better than GOBI in terms of objective scores. The interactive training with simulations allows GOBI* to converge quickly and adapt to volatile environments.

The convergence plots for GOBI* are shown in Figure~\ref{fig:convergence}. It demonstrates that as the model gets trained, the probability of the GOBI*'s actions being better than GOBI increases. This means that initially $\mathcal{L}_2$ has a high contribution to the loss for model training which aids GOBI* to quickly converge to GOBI's performance. This can be seen by the sudden drop in $\mathcal{L}_1$ in Figure~\ref{fig:convergence1}. This is also shown by Figure~\ref{fig:convegence2} that GOBI*'s loss converges faster with the additional $\mathcal{L}_2$ component. 

\begin{figure}
    \centering
    \subfigure[Components of GOBI* loss.]{
    \includegraphics[height=.26\linewidth]{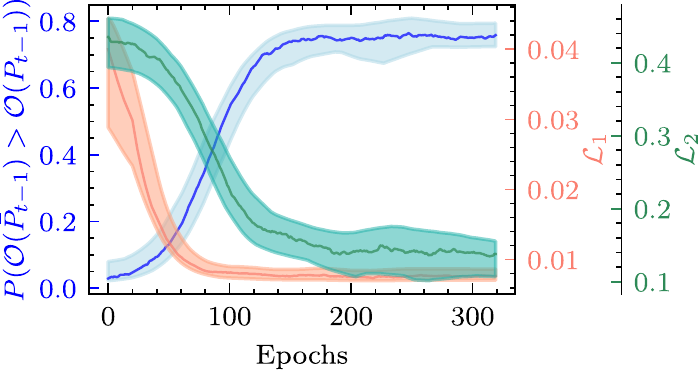}
    \label{fig:convergence1}
    }
    \subfigure[Comparison without $\mathcal{L}_2$]{
    \includegraphics[height=.26\linewidth]{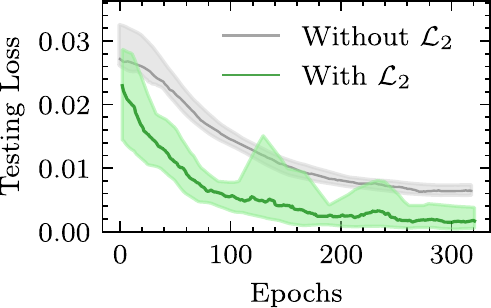}
    \label{fig:convegence2}
    }
    \caption{Convergence plots for GOBI*.}
    \label{fig:convergence}
\end{figure}

\section{Performance Evaluation}
\label{sec:perf_eval}

\subsection{Experimental Setup}
\label{sec:setup}
To test the efficacy of the proposed approaches and compare against the baseline methods, we perform experiments on both simulated and physical platforms. As COSCO has the same underlying models of task migration, workload generation and utilization metric values for both simulation and physical test-bench, we can test all models on both environments with the same underlying assumptions. As in prior work~\cite{ tuli2020healthfog}, we use $\alpha = \beta = 0.5$ in (\ref{eq:objective_function}) for our experiments. Moreover, we consider all tasks, to be allocated and migrated, in batches with each batch having feature vectors corresponding to up to $M = |H|^2$ tasks. If number of tasks is less than $M$, we pad the input matrix with zero vectors. This style of inference strategy is in line with prior work~\cite{tuli2020dynamic}. This is due to the limitation of non-recurrent neural networks to take a fixed size input for a single forward pass, hence requiring action inference for tasks being executed in batches.

\textbf{Physical Environment:}
We use the Microsoft Azure cloud provisioning platform to create a test-bed of 10 VMs located in two geographically distant locations. The gateway devices are considered to be in the same LAN of the server which was hosted in London, United Kingdom. Of the 10 VMs, 6 were hosted in London and 4 in Virginia, United States. The host capacities $\mathcal{C}(h_i) \forall h_i \in H$ are shown in Table~\ref{tab:hosts} and Table~\ref{tab:host_powers}. All machines use Intel Haswell 2.4 GHz E5-2673 v3 processor cores. The cost per hour (in US Dollar) is calculated based on the costs of similar configuration machines offered by Microsoft Azure at the datacenter in the South UK region\footnote{Microsoft Azure pricing calculator for South UK \texttt{https://azure.microsoft.com/en-gb/pricing/calculator/}}. The power consumption models are taken from the SPEC benchmarks repository\footnote{More details in Appendix~D}. We run all experiments for 100 scheduling intervals, with each interval being 300 seconds long, giving a total experiment time of 8 hours 20 minutes. We average over 5 runs and use diverse workload types to ensure statistical significance.

\begin{table*}[]
    \centering
    \caption{Host characteristics of Azure fog environment.}
    \resizebox{\textwidth}{!}{
        \begin{tabular}{@{}lcccccccccc@{}}
    \toprule 
    \multirow{2}{*}{Name} & \multirow{2}{*}{Quantity} & Core & \multirow{2}{*}{MIPS} & \multirow{2}{*}{RAM} & RAM & Ping & Network & Disk & Cost & \multirow{2}{*}{Location}\tabularnewline
     &  & count &  &  & Bandwidth & time & Bandwidth & Bandwidth & Model & \tabularnewline
    \midrule
    \multicolumn{11}{c}{\textbf{Edge Layer}}\tabularnewline
    \midrule
    Azure B2s server & 4 & 2 & 4029 & 4295 MB & 372 MB/s & 3 ms & 1000 MB/s & 13.4 MB/s & 0.0472 \$/hr & London, UK\tabularnewline

    Azure B4ms server & 2 & 4 & 8102 & 17180 MB & 360 MB/s & 3 ms & 1000 MB/s & 10.3 MB/s & 0.1890 \$/hr & London, UK\tabularnewline
    \midrule
    \multicolumn{11}{c}{\textbf{Cloud Layer}}\tabularnewline
    \midrule
    Azure B4ms server & 2 & 4 & 8102 & 17180 MB & 360 MB/s & 76 ms & 1000 MB/s & 10.3 MB/s & 0.166 \$/hr & Virginia, USA\tabularnewline

    Azure B8ms server & 2 & 8 & 2000 & 34360 MB & 376 MB/s & 76 ms & 2500 MB/s & 11.64 MB/s & 0.333 \$/hr & Virginia, USA\tabularnewline
    \bottomrule 
    \end{tabular}}
    \label{tab:hosts}
\end{table*}

\textbf{Simulation Environment:}
We consider 50 hosts machines as a scaled up version of the 10 machines from the last subsection. Here, each category has 5 times the instance count to give a total of 50 machines in a comparatively larger-scale fog environment as considered in prior art~\cite{basu2019learn, ahmed2018docker}. As we cannot place simulated nodes in geographically distant locations, we model the latency and networking characteristics of these nodes in our simulator as per Table~\ref{tab:hosts}.


\subsection{Model Training and Assumptions}

For the GOBI and GOBI* algorithms we use standard feed-forward neural models with the following characteristics adapted from \cite{tuli2020dynamic}. We use non-affine activation functions for our neural approximators to be differentiable for all input values.
\begin{itemize}[leftmargin=*]
    \item Input layer of size $M\times F + |H|\times F' + M\times N$ for GOBI and $M\times F + |H|\times F' + M\times N + 1$ for GOBI*. The non-linearity used here is $\textsf{softplus}$\footnote{The definitions of these activation functions can be seen at the PyTorch web-page: \texttt{https://pytorch.org/docs/stable/nn.html}} as in \cite{tuli2020dynamic}. Note, $|H|$ may vary from 10 in tests in the physical environment to 50 in the simulator.
    \item Fully connected layer of size $128$ with $\textsf{softplus}$ activation.
    \item Fully connected layer of size $64$ with $\textsf{tanhshrink}$ activation.
    \item Fully connected layer of size $1$ with $\textsf{sigmoid}$ activation.
\end{itemize}

To implement the proposed approach, we use PyTorch Autograd package~\cite{paszke2017automatic} to calculate the gradients of the network output with respect to input keeping model parameters constant. We generate training data for the GOBI model $f$ by running a random scheduler for 2000 intervals on the simulator. We then run the random scheduler on the framework for 300 intervals and fine-tune the GOBI model on this new dataset. For GOBI*, we use a random scheduler and predictions of the GOBI model $f$ to generate data of the form described in Section~\ref{sec:gobi} of size 2000 on simulator and fine-tune it on data corresponding to 300 intervals on the framework. The dataset was used to create the matrices $[\phi(A_{t-1}), \phi(H_{t-1}), \phi(\mathcal{D})]$, where each column of $\phi(A_{t-1}), \phi(H_{t-1})$ is normalized by the maximum and minimum utilization values, and $\phi(\mathcal{D})$'s rows are one-hot decision vectors. To train a model, we use the AdamW optimizer~\cite{saleh2019dynamic} with learning rate $10^{-5}$ and randomly sample 80\% of data to get the training set and the rest as the cross-validation set\footnote{All model training and experiments were performed on a system with configuration: Intel i7-10700K CPU, 64GB RAM, Nvidia GTX 1060 and Windows 10 Pro OS.}. The convergence criterium used for model training is the loss of consecutive epochs summed over the last 10 epochs is less than $10^{-2}$. For GOBI*, we also train predictors for utilization metrics. Specifically, we use a LSTM neural network~\cite{hochreiter1997long} with a single LSTM cell to achieve this and train it using the same dataset used to train the GOBI model.

\subsection{Workloads}

To generate workloads for training the GOBI and GOBI* models and testing them against baseline methods, we use two workload characteristics \textit{Bitbrain traces} and \textit{DeFog applications}. These were chosen because of their non-stationarity, their highly volatile workloads and their similarity with many real-world applications.
\begin{enumerate}[leftmargin=*]
    \item \textit{BitBrain traces}. The dynamic workload is generated for cloudlets based on the real-world, open-source Bitbrain dataset \cite{shen2015statisticalBitBrain}\footnote{The BitBrain dataset can be downloaded from: \texttt{http://gwa.ewi.tudelft.nl/datasets/gwa-t-12-bitbrains}}. This dataset consists of real traces of resource utilization metrics from 1750 VMs running on BitBrain distributed datacenter. The workloads running on these servers are from a variety of industry applications including computational analytical programs used by major banks, credit operators and insurers~\cite{shen2015statisticalBitBrain} and are commonly used for benchmarking fog-cloud models~\cite{tuli2020dynamic, khamse2018efficient, li2017bayesian}. The dataset consists of workload information for each time-stamp (separated by 5 minutes) including the number of requested CPU cores, CPU usage in terms of Million Instructions per Second (MIPS), RAM requested with Network (receive/transmit) and Disk (read/write) bandwidth characteristics. These different categories of workload data constitute the feature values of $\phi(A_{t-1})$ and $\phi(H_{t-1})$. As these workload traces correspond to applications running on real infrastructure, we use these to run our simulations and generate training data. At the start of each interval $I_t$, the size of the new tasks $N_t$ set follows a discrete Poisson distribution\footnote{The Poisson distribution models the number of independent arrivals that occur in a period, so it is apt to model for a batch size if the batch is formed by arrivals in a time window (scheduling interval in our case).} $Poisson(\lambda)$, as per prior works~\cite{tuli2020dynamic, he2021sla}. Here, $\lambda = 1.2$ jobs for $|H|= 10$ and $\lambda = 5$ jobs for $|H| = 50$.
    \item \textit{DeFog applications.} DeFog~\cite{mcchesney2019defog} is a fog computing benchmark which consists of six real-time heterogeneous workloads such as Yolo, Pocketspinx, Aeneas, FogLamp and iPokeMon. We use three specific heterogeneous applications of DeFog: Yolo (Memory, Bandwidth and Compute Intensive benchmark), PockeSphinx (Compute Intensive) and Aeneas (Bandwidth Intensive). Yolo runs with 1500 user requests and each request uploads a unique image with different density and can run upto three intervals. PochetSpinhx runs with 320 user requests and Aneases with 1500. Here too, at the start of each interval $I_t$, the size of the new tasks $N_t$ set follows a Poisson distribution $Poisson(\lambda)$ with $\lambda=1.2$ jobs.  Out of the $N_t$ tasks, the distribution of selection of Yolo/Pocketsphinx/Aeneas follows the probabilities of $(p_y, p_p, p_a)$ for a run. For our final comparison experiments we run using $p_y = p_p = p_a = 0.33$. However, to test diverse workload characteristics, we also run using any one as $0.80$ and other two as $0.10$. We call these workload dominant runs. We perform all our experiments to compare with baselines using the DeFog benchmarks.
\end{enumerate}

\subsection{Baseline Models}
\label{sec:baselines}
We evaluate the performance of the proposed algorithms and compare them against the state-of-the-art scheduling approaches. The reasons for comparing against these baselines are described in Section~\ref{sec:related-work}. We consider two light-weight heuristic based schedulers \textit{LR-MMT} and \textit{MAD-MC} that have low scheduling times, but higher energy consumption, response times and SLO violations. We also consider a deep-learning based gradient-free optimization method \textit{GA}, and two reinforcement-learning based schedulers \textit{DQLCM} and \textit{A3C}. Finally, we compare against a one max-weight based allocation method \textit{POND}.

\begin{figure*}
    \centering
    \subfigure[Average Energy Consumption]{
    \includegraphics[width=.23\textwidth]{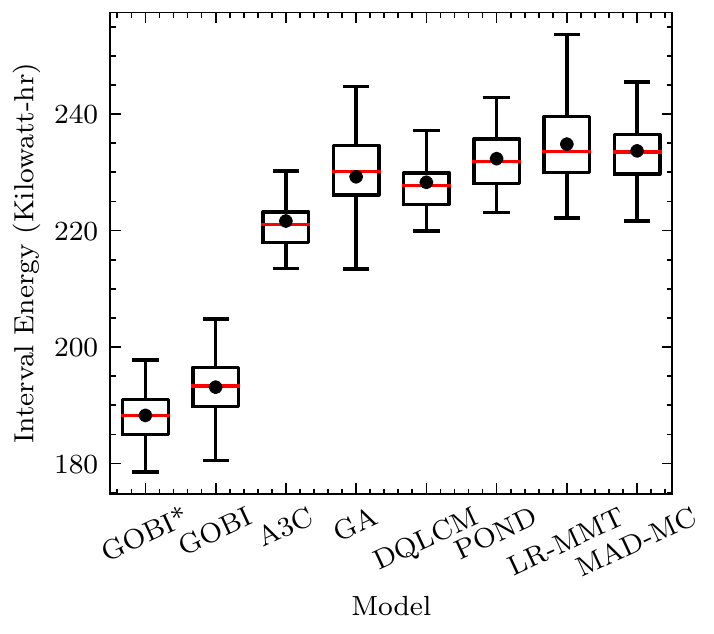}
    \label{fig:f_energy}
    }
    \subfigure[Average Response Time]{
    \includegraphics[width=.23\textwidth]{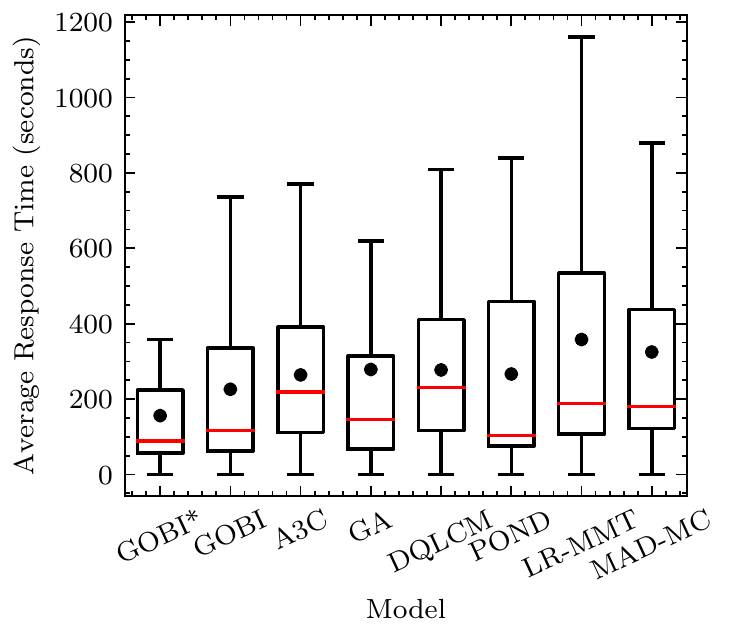}
    \label{fig:f_response}
    }
    \subfigure[Average Execution Time]{
    \includegraphics[width=.23\textwidth]{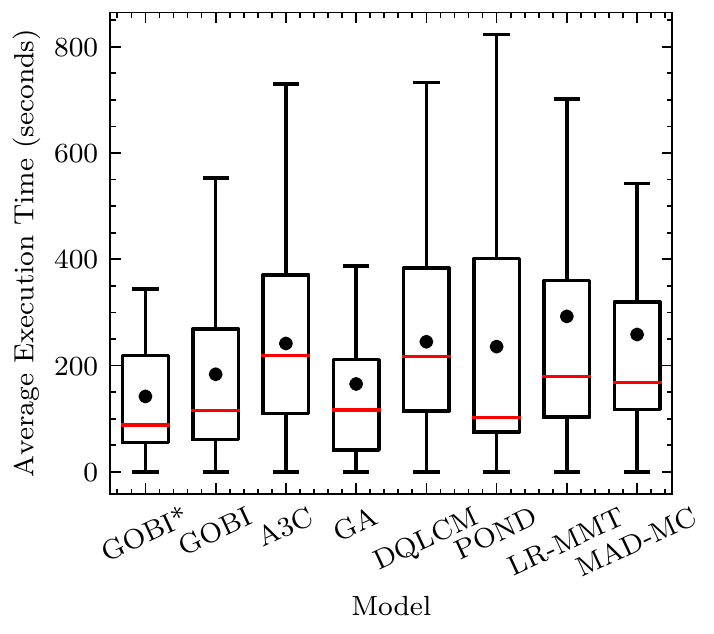}
    \label{fig:f_exec}
    }
    \subfigure[Fraction of SLO Violations]{
    \includegraphics[width=.23\textwidth]{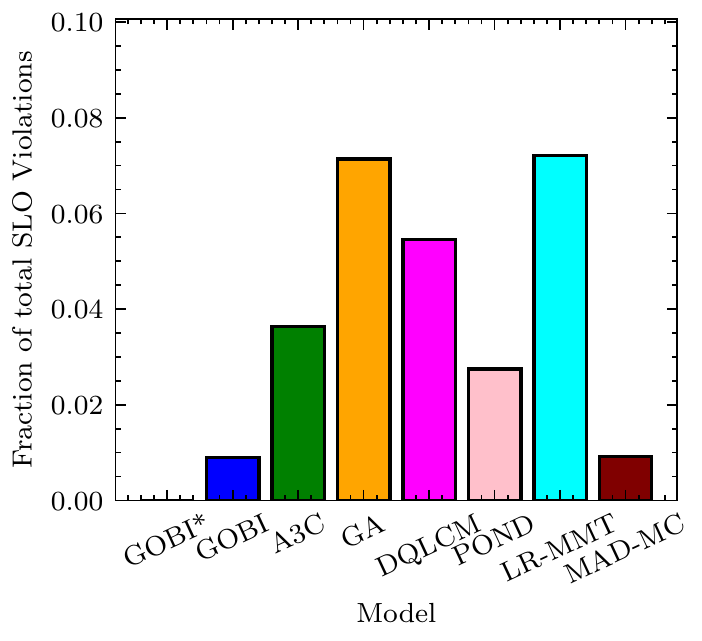}
    \label{fig:f_sla}
    }\\
    \subfigure[Average Migration Time]{
    \includegraphics[width=.23\textwidth]{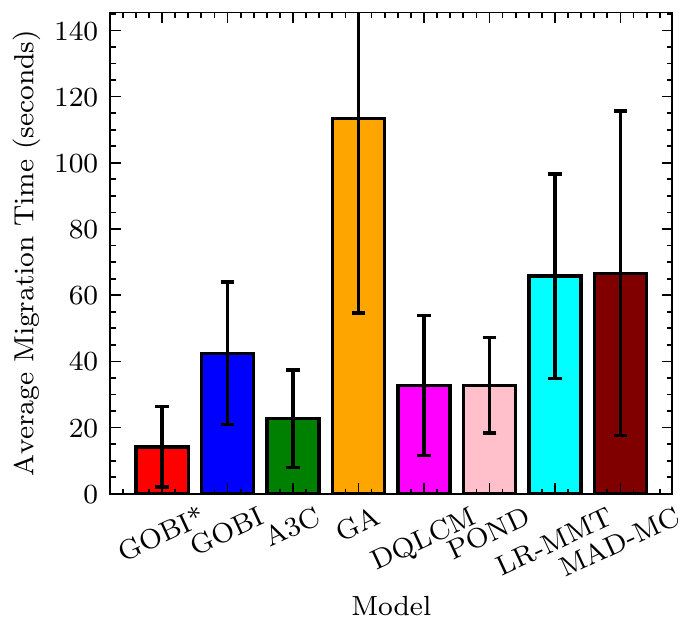}
    \label{fig:f_migration_time}
    }
    \subfigure[Average Response Time (per application)]{
    \includegraphics[width=.23\textwidth]{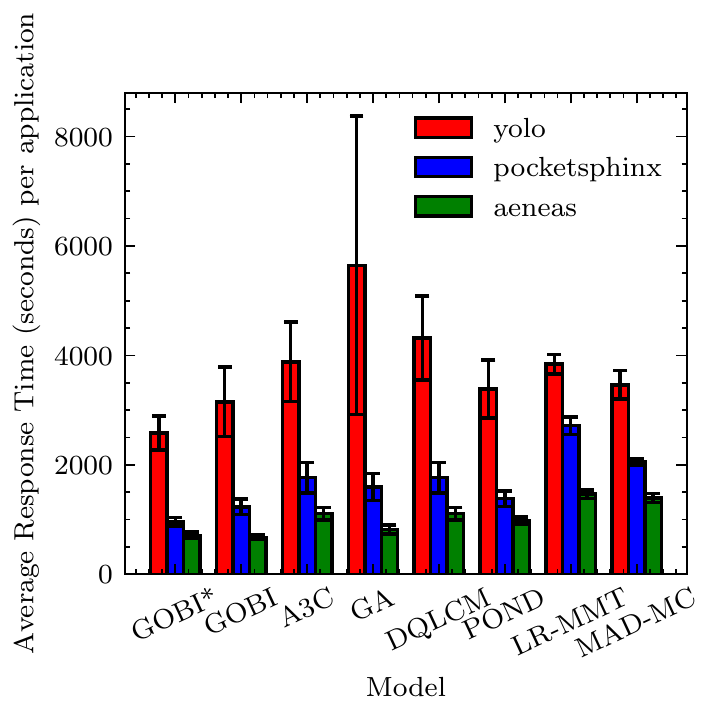}
    \label{fig:f_response_pa}
    }
    \subfigure[Average Wait Time (per application)]{
    \includegraphics[width=.23\textwidth]{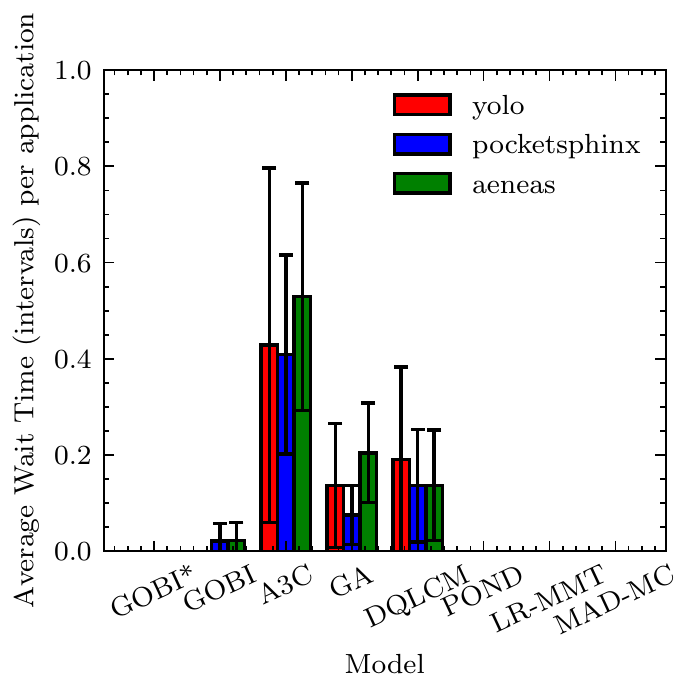}
    \label{fig:f_wait_pa}
    }
    \subfigure[Average SLO Violations (per application)]{
    \includegraphics[width=.23\textwidth]{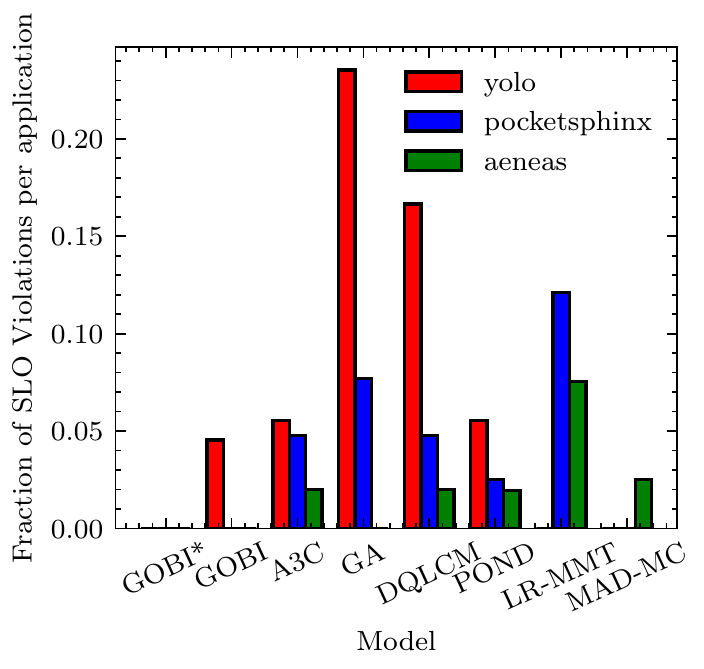}
    \label{fig:f_sla_pa}
    }\\
    \subfigure[Scheduling Time versus execution time]{
    \includegraphics[width=.23\textwidth]{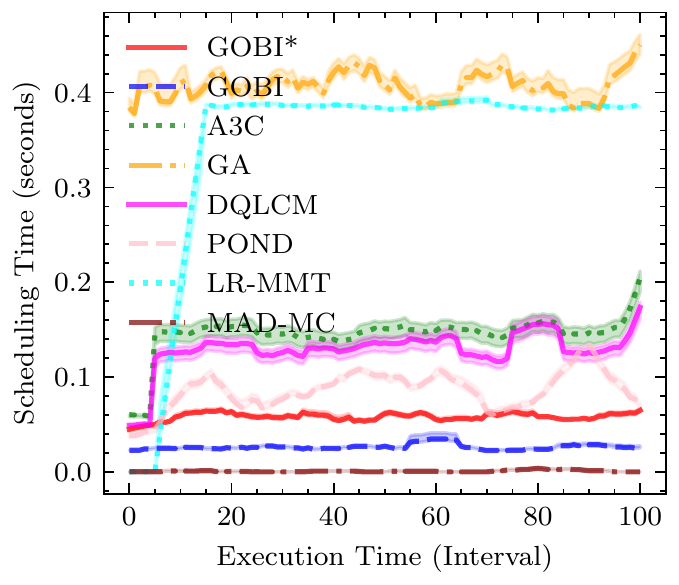}
    \label{fig:f_scheduling_time_series}
    }
    \subfigure[Average Scheduling Time]{
    \includegraphics[width=.23\textwidth]{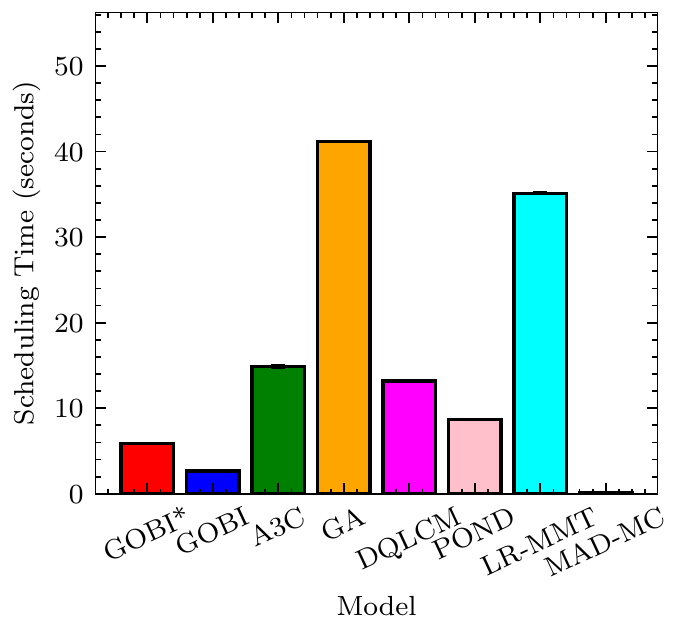}
    \label{fig:f_scheduling_time}
    }
    \subfigure[Average Wait Time]{
    \includegraphics[width=.23\textwidth]{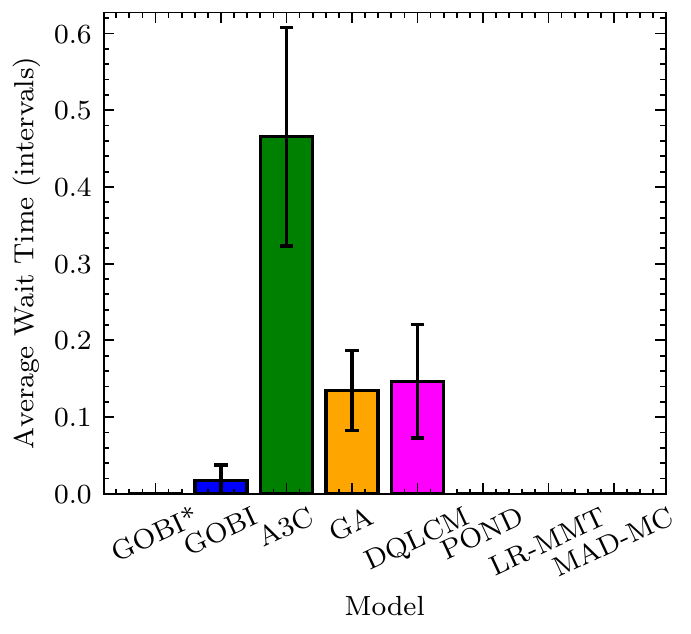}
    \label{fig:f_wait_time}
    }
    \subfigure[Fairness]{
    \includegraphics[width=.23\textwidth]{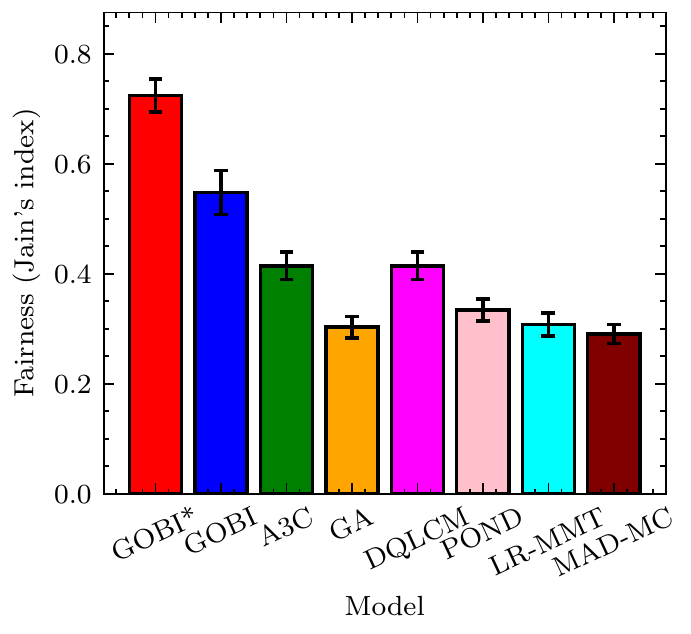}
    \label{fig:f_fairness}
    }
    \caption{Comparison of GOBI and GOBI* against baselines on framework with 10 hosts}
    \label{fig:framework_results}
\end{figure*}

\begin{figure*}
    \centering
    \subfigure[Average Energy Consumption]{
    \includegraphics[width=.23\textwidth]{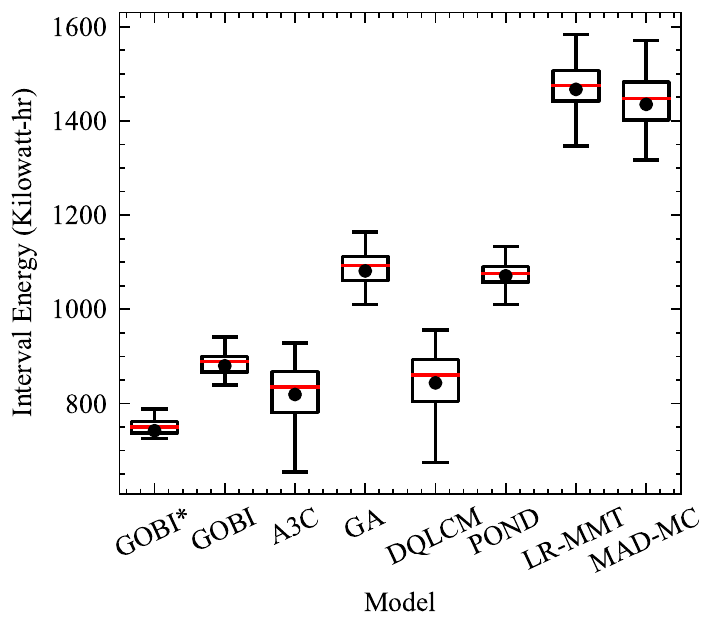}
    \label{fig:s_energy}
    }
    \subfigure[Average Response Time]{
    \includegraphics[width=.23\textwidth]{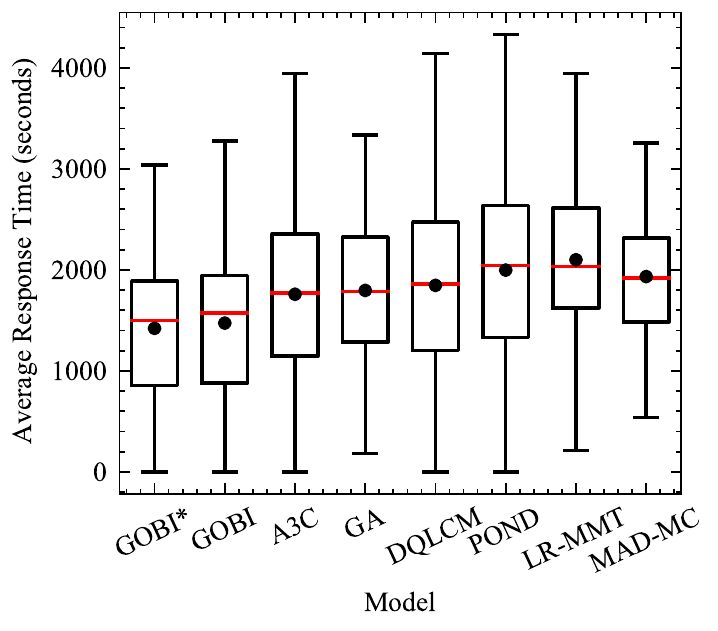}
    \label{fig:s_response}
    }
    \subfigure[Execution Time]{
    \includegraphics[width=.23\textwidth]{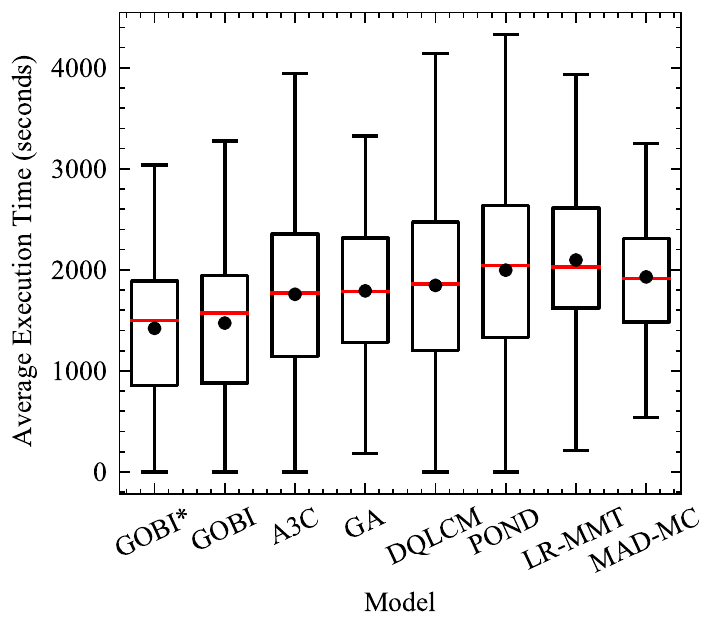}
    \label{fig:s_exec}
    }
    \subfigure[Migration Time]{
    \includegraphics[width=.23\textwidth]{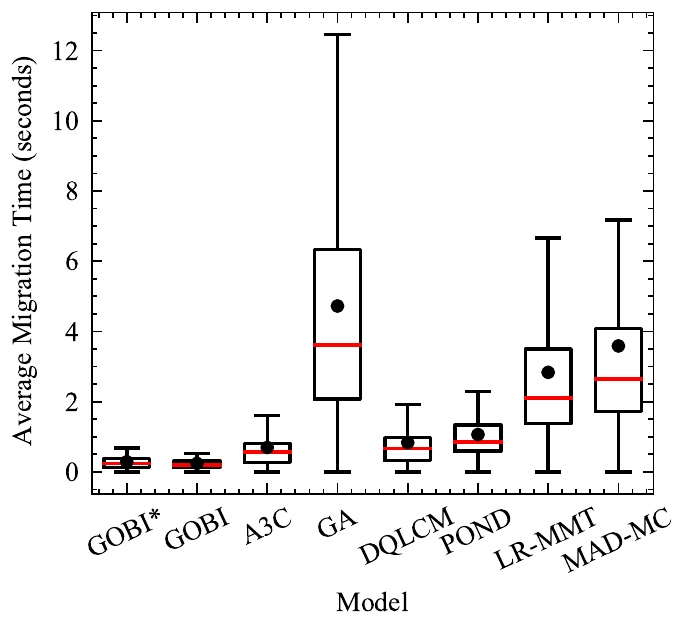}
    \label{fig:s_migration_time}
    }\\
    \subfigure[Fairness]{
    \includegraphics[width=.23\textwidth]{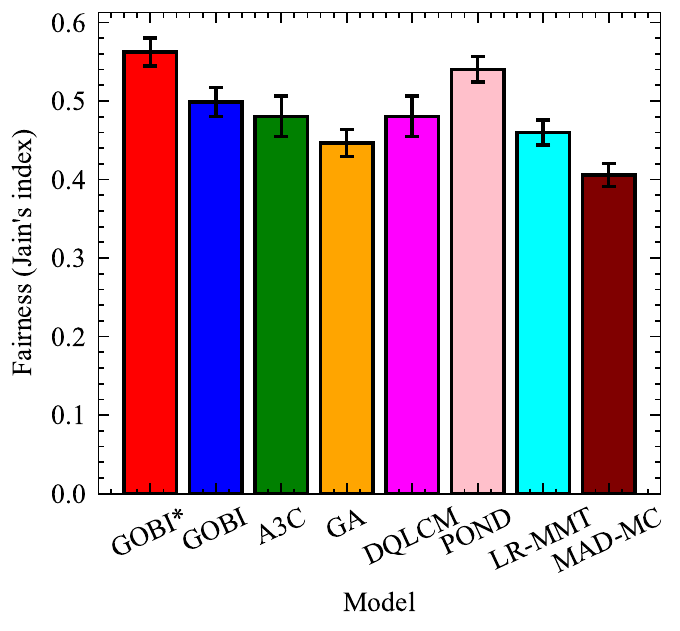}
    \label{fig:s_fairness}
    }
    \subfigure[Fraction of SLO Violations]{
    \includegraphics[width=.23\textwidth]{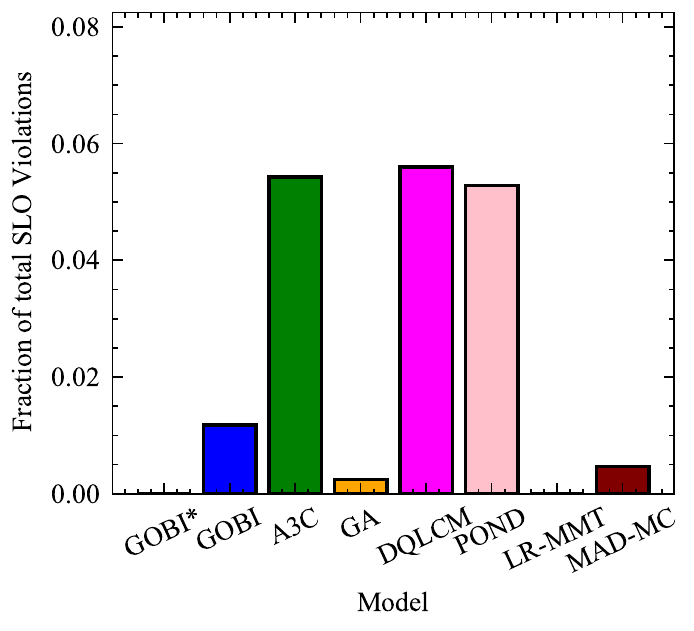}
    \label{fig:s_sla}
    }
    \subfigure[Average Scheduling Time]{
    \includegraphics[width=.23\textwidth]{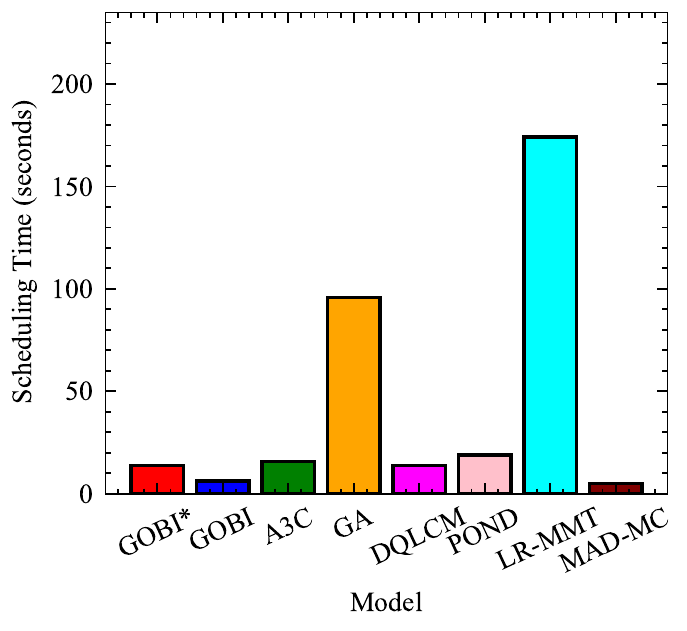}
    \label{fig:s_scheduling_time}
    }
    \subfigure[Average Wait Time]{
    \includegraphics[width=.23\textwidth]{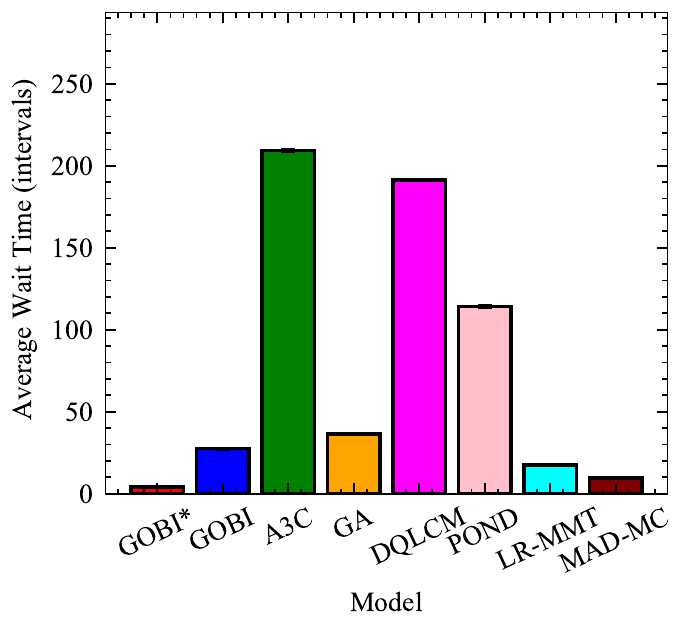}
    \label{fig:s_wait_time}
    }
    \caption{Comparison of GOBI and GOBI* against baselines on simulator with 50 hosts}
    \label{fig:simulator_results}
\end{figure*}

\subsection{Evaluation Metrics}
\label{sec:metrics}
We use the following evaluation metrics to test the GOBI and GOBI* models as motivated from prior works~\cite{basu2019learn, tuli2020dynamic, gill2019transformative, tuli2019fogbus, beloglazov2012optimal}. We also use $AEC$ and $ART$ as in Section~\ref{sec:gobi}.
\begin{enumerate}[leftmargin=*]
    \item \textit{SLO Violations} which is given as  \[\frac{\sum_t \sum_{l^t_j \in L_t} \mathbbm{1}(Response\ Time(l^t_j) \leq \psi(l^t_j))}{\sum_t |L_t|},\] where $\psi(l^t_j)$ is the 95$^{th}$ percentile response time for this application type (Yolo/PocketSphinx/Aeneas for DeFog and random/sequential for Bitbrain) on the state of the art baseline \textit{A3C}. This definition of percentile-based SLO is defined for the response time metrics of completed tasks and is inspired from~\cite{boloor2010dynamic}.
    \item \textit{Fairness} which is given as the Jain's Fairness Index \[\frac{( \sum_t \sum_{l^t_j\in L_t} Response\ Time(l^t_j))^2}{(\sum_t |L_t|) \times ( \sum_t \sum_{l^t_j\in L_t} Response\ Time(l^t_j)^2 )}.\]
    \item \textit{Average Migration Time} which is the average time for the migration of tasks over all intervals in the run \[\frac{\sum_t \sum_{(a_{j}^t, h_{i}) \in \hat{\mathcal{D}}(A_t)} Migration\ Time(a_{j}^t, h_{i})}{\sum_t |\hat{\mathcal{D}}(A_t)|}.\]
    \item \textit{Scheduling Time} which is the average time to reach the scheduling decision over all intervals in the run.
    \item \textit{Average Wait Time} which is the average time for a container in the wait queue before it starts execution.
\end{enumerate}


\subsection{Results}

In this section we provide comparative results, showing how GOBI and GOBI* perform against other baselines as described in Section~\ref{sec:baselines}. We compare the proposed approaches with respect to the evaluation metrics described in Section~\ref{sec:metrics}. Additional results are given in Appendix~E. The graphs in Figure~\ref{fig:framework_results} show the results for runs with length of 100 scheduling intervals viz 8 hours 20 minutes using the DeFog workloads on 10 physical Azure machines. Figure~\ref{fig:simulator_results} shows similar trends in the result for 50 hosts in a simulated environment.

Figure~\ref{fig:f_energy} shows the energy consumption of Azure host machines for each scheduling policy. Among the baseline approaches, \textit{A3C} consumes the least average interval energy of $221.63$ KW-hr. Running the \textit{GOBI} approach on the Azure platform consumes $193.11$ KW-hr, $12.86\%$ lower than \textit{A3C}. Further, \textit{GOBI*} consumes the least energy $188.26$ KW-hr ($15.05\%$ lower than that of \textit{A3C}). The major reason for this is the low response time of each application, which leads to more tasks being completed within the same duration of 100 intervals. Thus, even with similar energy consumption across all policies, \textit{GOBI} and \textit{GOBI*} have a much higher task completion rate of 112-115 tasks in 100 intervals compared with only 107 tasks. Similarly, when we consider 50 host machines in a simulated platform (Figure~\ref{fig:s_energy}), \textit{GOBI*} still consumes $9.39\%$ lower than \textit{A3C}. The energy consumption with \textit{GOBI} is very close to \textit{A3C}.

Figure~\ref{fig:f_response} shows the average response time for each policy. Here, response time is the time between the creation of a task from an IoT sensor and up to the gateway receiving the response. Among the baselines models, \textit{POND} has the lowest average response time of $266.53$ seconds. Here, \textit{GOBI} and \textit{GOBI*} have $226.04$ and $156.09$ seconds ($15.19\%-41.43\%$ better than \textit{POND}). For the three applications, the average response time for each policy is shown in Figure~\ref{fig:f_response_pa}. Clearly, among the DeFog benchmark applications, Yolo takes the maximum time.  The highest response time for Yolo among all policies is for \textit{GA} as it prefers scheduling shorter jobs first \emph{i.e.} Pocketsphinx and Aeneas to reduce response time. This is verified by the results which show that Aeneas has the lowest response time when the \textit{GA} policy is used. Another disadvantage of this preference is that it leads to high wait times when running the \textit{GA} policy. However, in larger scale experiments on 50 hosts in simulation (Figure~\ref{fig:s_response}), \textit{GA} has much lower wait times leading to the best response times among the baseline algorithms. Here, compared to \textit{GA}, \textit{GOBI} and \textit{GOBI*} have $17.98\%-20.87\%$ lower average response times, respectively.

Figure~\ref{fig:f_wait_time} shows the average waiting time (in intervals) for tasks running on the Azure framework with each policy. As prior works have established that RL approaches are slow to adapt in highly volatile environments~\cite{tuli2020dynamic, findling2019computational}, the \textit{A3C} and \textit{DQLCM} approaches are unable to adapt when a host is running at capacity. This means that for a large number of scheduling intervals ($47.72\%$), they predict an overloaded host. Hence, the task cannot be assigned in the same interval and has to wait until either it is assigned to another host, or the resources of this host are released. This is also reflected in the wait time per application (Figure~\ref{fig:f_wait_pa}). As established in prior work, reinforcement-learning approaches scale poorly and hence this waiting time gets more pronounced when running with 50 hosts (Figure~\ref{fig:s_wait_time}). However, as \textit{GOBI} and \textit{GOBI*} are able to adapt quickly to the environment, because of the neural model update in each iteration, they do not face such problems even on larger scale experiments.
 
As seen in Figure~\ref{fig:f_sla}, due to the high response times of the baseline approaches, they have a high SLO violation rate. For the experiments of physical hosts, as \textit{POND} has the lowest response time, the SLO violation rate is also the lowest among the baselines ($2.75\%$). \textit{GOBI} has only $0.9\%$ SLO violations and \textit{GOBI*} has none. This is due to the back-propagation approach which minimizes the response time, as well as avoiding local optima using the adaptive moment estimation approach (Adam optimizer). For \textit{GOBI}, the SLO violations are only for Yolo tasks ($4.5\%$). In the simulation environment, due to the high wait and response times of the \textit{A3C}, \textit{DQLCM} and \textit{POND} schedulers, their SLO violation rates are highest among all methods. Here, \textit{GOBI} has only $1.1\%$ SLO violations. Due to the high response times of \textit{A3C} and \textit{DQLCM} for 50 hosts, the SLO itself is much higher giving the \textit{GA} approach with low response times only $0.2\%$ SLO violations. \textit{GOBI*} has no SLO violations in this case as well.

The main reason for low response times in \textit{GOBI} and \textit{GOBI*} is their low scheduling time. Even with single-step simulation, the back-propagation based optimization strategy with the AdamW optimizer \textit{converges faster than the gradient-free optimization} methods like \textit{A3C}, \textit{DQLCM} \textit{GA} and \textit{POND}. This is apparent from the scheduling times seen in Figures~\ref{fig:f_scheduling_time} and \ref{fig:f_scheduling_time_series} (shaded region shows the 95\% confidence interval). Here, \textit{MAD-MC} has the lowest scheduling time, with \textit{GOBI} next at $2.65$ seconds and \textit{GOBI*} at $5.88$ seconds. Even with 50 hosts (Figure~\ref{fig:s_scheduling_time}), \textit{GOBI} and \textit{GOBI*} have one of the lowest scheduling times across all policies (except \textit{MAD-MC}).

Figure~\ref{fig:f_migration_time} shows the average migration times for all polices. Here, \textit{GA} has the highest migration time due to the largest number of migrations. This is as expected due to the \textit{non-local jumps} in the \textit{GA} approach which lead to a high number of migrations. Approaches like the back-propagation method (\textit{GOBI} and \textit{GOBI*}) and RL based approaches (\textit{A3C} and \textit{DQLCM}) have comparatively low migration times. A similar trend is also apparent with 50 hosts as seen in Figure~\ref{fig:s_migration_time}. Due to convergence to better optima, the execution times itself for iterative optimization approaches including \textit{GOBI}, \textit{GOBI*} and \textit{GA} are the lowest among all approaches (Figures~\ref{fig:f_exec} and \ref{fig:s_exec}). Finally, Figures~\ref{fig:f_fairness} and \ref{fig:s_fairness} show that \textit{GOBI} and \textit{GOBI*} are the most fair schedulers.

\subsection{Comparing GOBI and GOBI*}

Now that \textit{GOBI} and \textit{GOBI*} have been shown to out-perform prior works, we directly compare the two approaches. Table~\ref{tab:comparison} shows \textit{GOBI*} has a lower objective value ($16.4\%$) and prediction error of objective value ($14.2\%$) for the next interval. This clearly shows that \textit{GOBI*} can optimize the objective value better than \textit{GOBI}. This is because \textit{GOBI} predicts the objective value for the interval $I_t$ using only the utilization metrics of interval $I_{t-1}$. However, \textit{GOBI*} has more information as the LSTM models predict utilization metrics for $I_t$ and the \textit{GOBI} model gives a tentative action, which is then used to simulate $I_t$ and get a reasonable estimate of the objective value for $I_t$. This estimate with the utilization metrics of $I_{t-1}$ aid \textit{GOBI*} to have much closer predictions and hence better scheduling decisions. However, the scheduling time of \textit{GOBI*} is more than twice that of \textit{GOBI}. This makes the \textit{GOBI} approach more appropriate for environments with resource constrained servers.

\begin{table}[]
    \centering
    \caption{Comparison between GOBI and GOBI*}
    \resizebox{\columnwidth}{!}{
    \begin{tabular}{@{}lccc@{}}
    \toprule 
    Algorithm & Scheduling Time & Objective Value & MSE\tabularnewline
    \midrule
    GOBI & \textbf{2.65$\pm$0.019 s} & 0.5673 & 2.31$\pm$0.15$\times10^{-3}$\tabularnewline

    GOBI{*} & 5.88$\pm$0.015 s & \textbf{0.4744} & \textbf{1.98$\boldsymbol\pm$0.22$\boldsymbol{\times10^{-3}}$}\tabularnewline
    \bottomrule 
    \end{tabular}}  
    \label{tab:comparison}
\end{table}

\section{Conclusions and Future Work}
\label{sec:conclusion}

We have presented a coupled-simulation approach to leverage simulators to predict the QoS parameters and make better decisions in a heterogeneous fog setup. The presented COSCO framework allows deployment of a holistic platform that provides an easy to use interface for scheduling policies to access simulation capabilities. Moreover, we have presented two scheduling policies based on back-propagation of gradients with respect to input, namely GOBI and GOBI*. These approaches use neural approximators to model objective scores for the scheduling decisions. GOBI* uses GOBI, predictors of utilization characteristics and an underlying simulator to better predict the objective values leading to better decisions. Comparing GOBI and GOBI* against state-of-the-art schedulers using real-world fog applications, we see that our methods are able to reduce energy consumption, response time, SLO violations and scheduling time. Between GOBI and GOBI*, GOBI* gives better QoS values; however, GOBI is more suitable for resource constrained servers.

For the future, we propose to extend the COSCO framework to allow workflow models for serverless computing. Extending to serverless would allow us to perform fine-grained auto-scaling, increase productivity and improve flexibility and logistics~\cite{casale2020radon}. For the back-propagation approach, we wish to extend our methods to consider layer types and activations like recurrent, convolution or residual with Rectified Linear Units (ReLU). This is because such non-differentiable functions are increasingly being used to approximate diverse objective functions. More advanced layer types would also allow us to model temporal and spatial characteristics of the environment. 

\section*{Software Availability}
\footnotesize{The code is available at \url{https://github.com/imperial-qore/COSCO}. The Docker images are available at \url{https://hub.docker.com/u/shreshthtuli}. The training datasets and execution traces are available at \url{https://doi.org/10.5281/zenodo.4897944}, released under the CC BY 4.0 license.}

\section*{Author Contributions}

\footnotesize{S.T. designed and implemented the COSCO framework with GOBI, and GOBI* algorithms. S.P aided in software debugging. S.T. designed and conducted the experiments. S.T., S.P., G.C and N.R.J analyzed the experimental results. S.T, G.C. and N.R.J. wrote the paper. G.C., S.N.S and N.R.J supervised the project.}

\section*{Acknowledgments}
\footnotesize{S.T. is grateful to the Imperial College London for funding his Ph.D. through the President’s Ph.D. Scholarship scheme. S.P. is supported by  the  European  Social Fund via IT Academy program. The work of S.T. and G.C. has been partly funded by the EU’s Horizon 2020 program under grant agreement No 825040.}

\appendices

\section{Implementation of COSCO framework}
\begin{figure*}[h]
    \centering
    \includegraphics[width=0.95\textwidth]{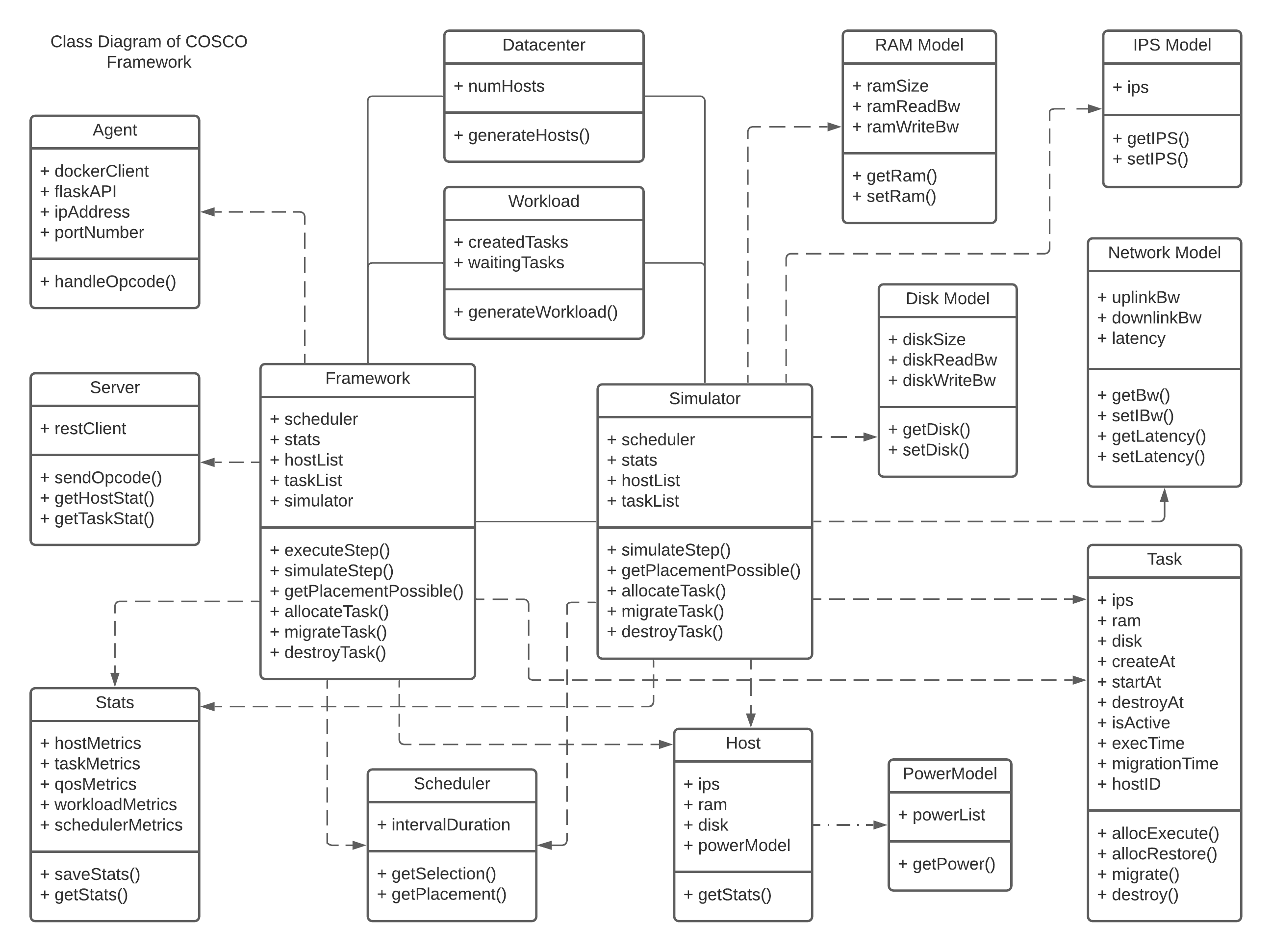}
    \caption{UML class diagram of the COSCO framework}
    \label{fig:class}
\end{figure*}

We present the implementation details of the simulator and framework with the interface details. To implement the COSCO framework, we use object-oriented programming model and realize it using the Python programming language. We now provide details of various Python classes which together consolidate to the complete framework. A UML class diagram of the COSCO framework is shown in Figure~\ref{fig:class}.

The hosts are instantiated/simulated using the \textit{Datacenter} class which provides a list of \textit{Host} objects to the simulator or framework. These \textit{Host} objects might correspond to simulated hosts or physical compute nodes with Instructions per second (IPS), RAM, Disk and Bandwidth capacities with communication latency and power consumption models with CPU utilization. The workloads are generated by the \textit{Workload} class. These are either time-series models of IPS, RAM, Disk, Bandwidth requirements in the case of simulation or are actual application programs in the case of physical experiments. At the beginning of each scheduling interval, a \textit{Workload} object creates a list of workloads ($N_t$) and also maintains a waiting queue for workloads that could not be allocated in a prior interval ($W_{t-1}$). Moreover, the \textit{Scheduler} class allows taking scheduling decisions in the form of task selection (\texttt{getSelection()}) and placement (\texttt{getPlacement()}). Task selection function selects the containers that need to be migrated to another host; task placement function returns the target host for each selected or new container. 

\subsection{Modelling of Computational Utilization}
The utilization metric indirectly used for our experiments is the \% CPU consumption, which is used as part of the workload and host characteristics (input to the neural approximators). However, as part of the host capacity characteristics, IPS is taken by running SPEC benchmark using the ``\texttt{perf stat}'' command on real systems to obtain an approximation of how many maximum number of instructions per second a host can execute. For host/workload utilization characteristics, using the \% CPU utilization, we get the \% CPU time allocated to a particular task for that interval. We also get the number of instructions executed using the ``\texttt{perf stat}'' command. This allows us to calculate an ``average'' number of instructions executed per second for each workload. Summing this IPS for each host over all running workloads on that host gives us the IPS of that host. Thus, the above-mentioned calculation in our physical experiments gives us a proxy to the IPS metric commonly referred to in simulated scenarios. This makes the IPS used in our model is agnostic to the underlying infrastructure (simulator/physical setup), being taken directly from the Bitbrain traces in our simulation experiments and being calculated from the \% CPU utilization in our physical experiments.

\subsection{Simulator}
\label{sec:simulator}
The \textit{Simulator} class consists of a \textit{Stats} object, \textit{Scheduler} object, lists of \textit{Host} and \textit{Task} objects. A \textit{Stats} object saves the utilization metrics of host and containers, allocation decisions, QoS parameters and workload characteristics for each scheduling interval. The \textit{Scheduler} object, as described earlier, uses the utilization metrics and QoS parameters to take an allocation decision for new workloads or migration decisions for active workloads. The migration time of a container is calculated by summing the networking latency with the time it takes to transfer the container based on its size and the network bandwidth. Naturally, if the container is on one of the edge or cloud layers and is migrated to the same layer, we assume no network latency, but if it is transferred across layers, we also include the latency.
At the end of each interval, the \textit{Stats} object calculates the QoS parameters like energy using the power models of each host and utilization metrics.   

\subsection{Framework}
\label{sec:framework}
The \textit{Framework} class consists of the objects in the \textit{Simulator} class with some additional objects for container orchestration. We use the Ansible framework~\cite{hochstein2017ansible} to autonomously provision, deploy and manage physical host machines in a fog environment. Each \textit{Task} object here is a Docker container. Docker is a container management platform as a service used to build, ship and run containers on physical or virtual environments~\cite{ahmed2018docker}. Moreover, the \textit{Host} objects of \textit{Framework} are physical computational nodes in the same Virtual Local Area Network (VLAN) as the server machine on which the Python program runs. The server communicates with each \textit{Host} via HTTP REST APIs~\cite{yates2015ensembl}. The REST server of the \textit{Framework} sends operation codes including \texttt{GetHostStats, GetContainerStats, Checkpoint, Migrate, Restore} to each host to instruct a specific activity. At every \textit{Host} machine, a \texttt{DockerClient} service runs with a \texttt{Flask} HTTP web-service~\cite{grinberg2018flask}. Flask is web application REST client which allows interfacing between \textit{Hosts} and the server. 

A \textit{Host} receives operation codes from the server and returns HTTP response. Before starting execution of containers, the server gets the utilization metric capacities from each host ($\mathcal{C}(h_i), \forall h_i \in H$). \texttt{GetHostStats} returns the host utilization metrics of host, \texttt{GetContainerStats} returns utilization metrics of each individual container running in the host. 
All migrations are performed in parallel using \texttt{Joblib} library in Python. The allocation and migrations are done in a non-blocking fashion, i.e. the execution of active containers is not halted when some tasks are being migrated. All utilization metrics are synchronized among devices via a time-series database, namely InfluxDB~\cite{naqvi2017time}.

\begin{table*}[]
    \centering
    \caption{Power consumption of Azure instances with percent CPU utilization.}
    \label{tab:host_powers}
    \begin{tabular}{@{}lccccccccccc@{}}
    \toprule 
    \multirow{2}{*}{Name} & \multicolumn{11}{c}{SPEC Power (Watts) for different CPU percentage usages}\tabularnewline
    \cmidrule{2-12}
     & 0\% & 10\% & 20\% & 30\% & 40\% & 50\% & 60\% & 70\% & 80\% & 90\% & 100\%\tabularnewline
    \midrule 
    Azure B2s server & 75.2 & 78.2 & 84.1 & 89.6 & 94.9 & 100.0 & 105.0 & 109.0 & 112.0 & 115.0 & 117.0\tabularnewline

    Azure B4ms server & 71.0 & 77.9 & 83.4 & 89.2 & 95.6 & 102.0 & 108.0 & 114.0 & 119.0 & 123.0 & 126.0\tabularnewline
    \midrule 
    Azure B4ms server & 71.0 & 77.9 & 83.4 & 89.2 & 95.6 & 102.0 & 108.0 & 114.0 & 119.0 & 123.0 & 126.0\tabularnewline

    Azure B8ms server & 68.7 & 78.3 & 84.0 & 88.4 & 92.5 & 97.3 & 104.0 & 111.0 & 121.0 & 131.0 & 137.0\tabularnewline
    \bottomrule 
    \end{tabular}
\end{table*}

\section{Gradient with respect to Input for other activations}

We present the gradients with respect to input for the $\textsf{softplus()}$ and $\textsf{tanhshrink()}$ activation functions.

\begin{theorem}[Gradient of $\textsf{softplus}()$]
\label{backpropinput}
For a linear layer with \textsf{{softplus}}() non-linearity that is defined as follows:
\[ f(x;W,b) = \textsf{softplus}(W\cdot x + b), \]
the derivative of this layer with respect to $x$ is given as:
\[ \nabla_x f(x;W,b) = W^T \times \textsf{sigmoid}(W\cdot x + b).\]
\end{theorem}
\begin{proof}
It is well known that $\nabla_x \textsf{sofptlus}(x) = \textsf{sigmoid}(x)$, therefore,
\begin{align*}
    \nabla_x f(x;W,b) &= \nabla_x \textsf{softplus}(W\cdot x + b)\\
    &= \nabla_{W\cdot x + b} \textsf{softplus}(W\cdot x + b) \times \nabla_x (W\cdot x + b)\\
    &= \textsf{sigmoid}(W\cdot x + b) \times \nabla_x (W\cdot x + b)\\
    &= W^T \times \textsf{sigmoid}(W\cdot x + b) \qedhere
\end{align*}
\end{proof}

\begin{theorem}[Gradient of $\textsf{tanhshrink}()$]
\label{backpropinput}
For a linear layer with \textsf{tanhshrink}() non-linearity that is defined as follows:
\[ f(x;W,b) = \textsf{tanhshrink}(W\cdot x + b), \]
the derivative of this layer with respect to $x$ is given as:
\[ \nabla_x f(x;W,b) = W^T \times \textsf{sigmoid}(W\cdot x + b).\]
\end{theorem}
\begin{proof}
As $\textsf{tanhshrink}(x) = x - \textsf{tanh}(x)$, thus
\begin{align*}
\nabla_x\textsf{tanhshrink}(x) &= 1 - \nabla_x \textsf{tanh}(x)\\
 &= 1 - (1 - \textsf{tanh}^2(x))\\
 &= \textsf{tanh}^2(x),
\end{align*}
therefore,
\begin{align*}
    \nabla_x f(x;W,b) &= \nabla_x \textsf{tanhshrink}(W\cdot x + b)\\
    =\ &\nabla_{W\cdot x + b} \textsf{tanhshrink}(W\cdot x + b) \times \nabla_x (W\cdot x + b)\\
    =\ &\textsf{tanh}^2(W\cdot x + b) \times \nabla_x (W\cdot x + b)\\
    =\ &W^T \times \textsf{tanh}^2(W\cdot x + b) \qedhere
\end{align*}
\end{proof}

\section{Training GOBI and GOBI* models}

\begin{figure}
    \centering
    \includegraphics[width=0.8\linewidth]{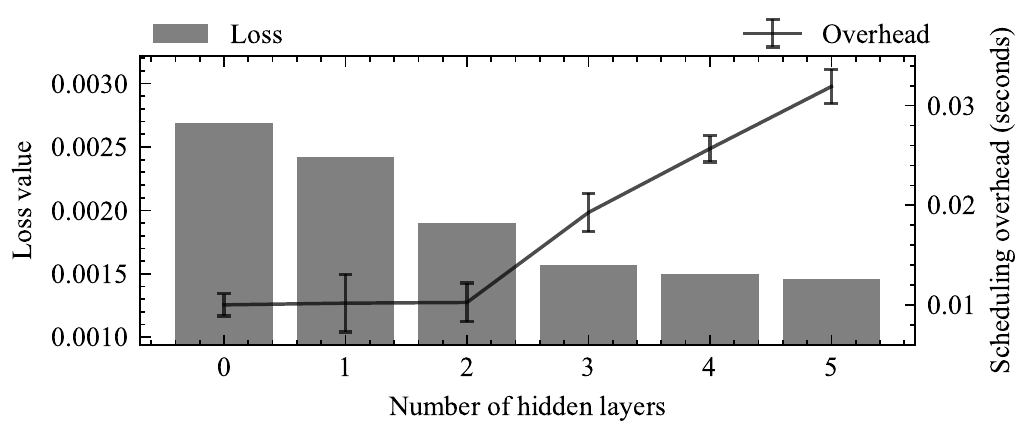}
    \caption{Loss and scheduling overhead as a function of the number of hidden layers.}
    \label{fig:layers}
\end{figure}

\subsection{Model search}

To determine the optimum number of layers in our feed-forward neural model, we analyze the scheduling overhead of the GOBI approach with the number of hidden layers. The scheduling overhead is calculated as the average time it takes to take a scheduling decision in our simulated setup. As shown in Figure~\ref{fig:layers}, the value of the loss function decreases with the increase in the number of layers of the neural network. This is expected because as the number of layers increase so do the number of parameters and thus the ability of the network to fit more complex functions becomes better. The scheduling overhead depends on the system on which the simulation is run (the fog broker), and for the current experiments, the system used had CPU - Intel i7-10700K and GPU - Nvidia GTX 1060 graphics card (4GB graphics RAM). As shown in the figure, there is an inflection point at 2 hidden layers because a network with 2 or more such layers could not fit with the input batch in the GPU graphics RAM and had to be run on CPU, increasing inference times. Based on the available simulation infrastructure, for the comparison with baseline algorithms, we use the neural network with 2 hidden layers based on the above described grid search approach.

\begin{figure}[t]
    \centering
    \subfigure[GOBI]{
    \includegraphics[width=.23\textwidth]{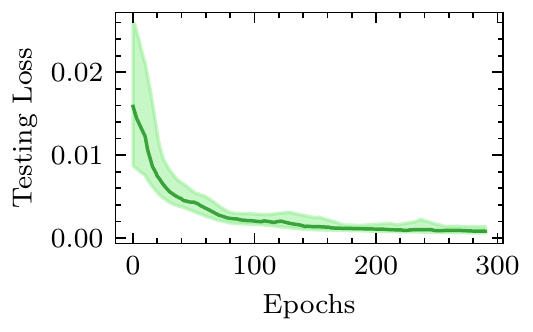}
    \label{fig:s_gobi}
    }
    \subfigure[GOBI*]{
    \includegraphics[width=.23\textwidth]{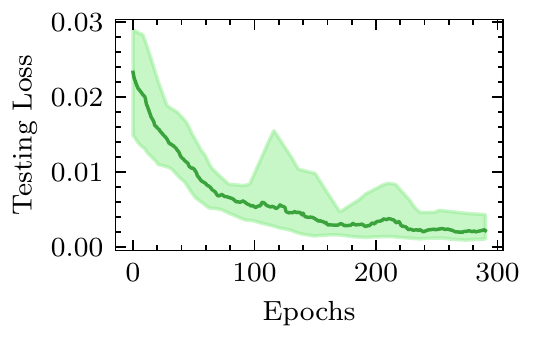}
    \label{fig:s_gobi2}
    }
    \caption{Training graphs for simulated setup with 50 nodes.}
    \label{fig:simulator_training}
\end{figure}

\begin{figure}[t]
    \centering
    \subfigure[GOBI]{
    \includegraphics[width=.23\textwidth]{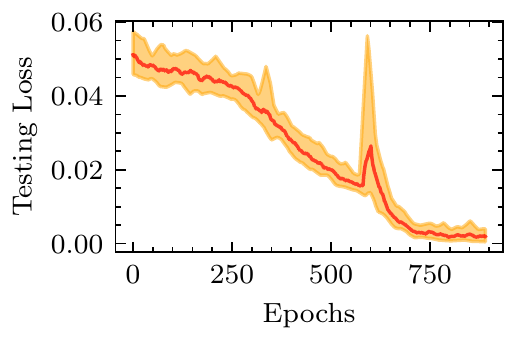}
    \label{fig:s_gobi}
    }
    \subfigure[GOBI*]{
    \includegraphics[width=.23\textwidth]{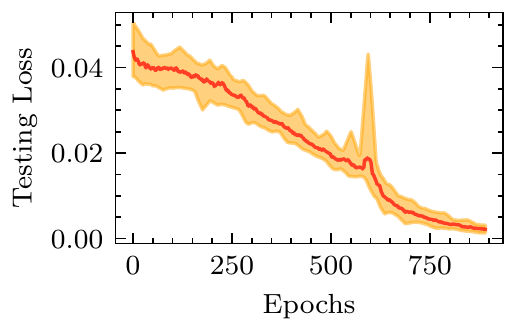}
    \label{fig:s_gobi2}
    }
    \caption{Training graphs for framework setup with 10 nodes.}
    \label{fig:framework_training}
\end{figure}

\begin{figure*}
    \centering
    \subfigure[Average CPU Utilization of hosts]{
    \includegraphics[width=.23\textwidth]{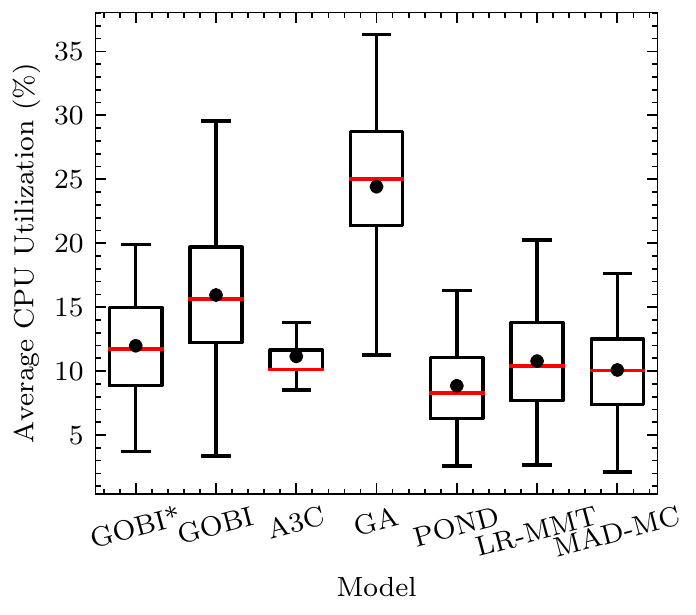}
    \label{fig:f_cpu}
    }
    \subfigure[Average Number of active containers per interval]{
    \includegraphics[width=.23\textwidth]{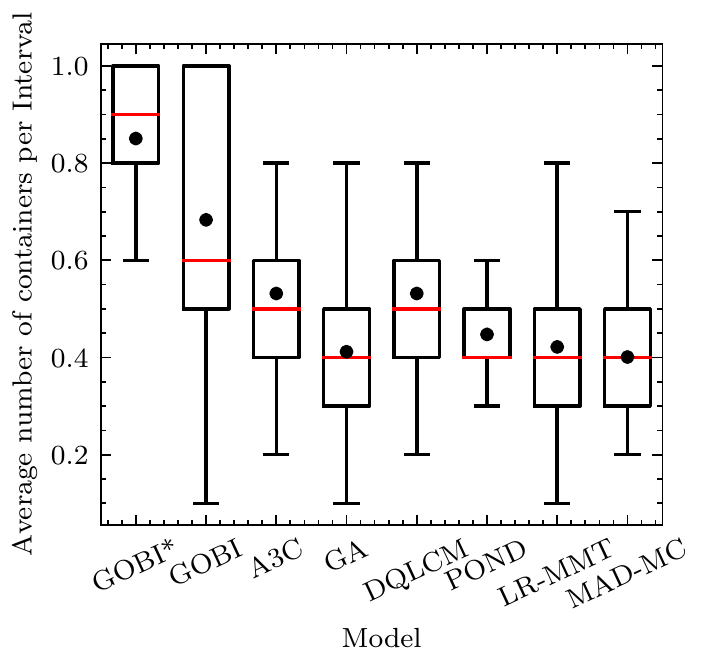}
    \label{fig:f_containers}
    }
    \subfigure[Average Cost per container (US Dollars)]{
    \includegraphics[width=.23\textwidth]{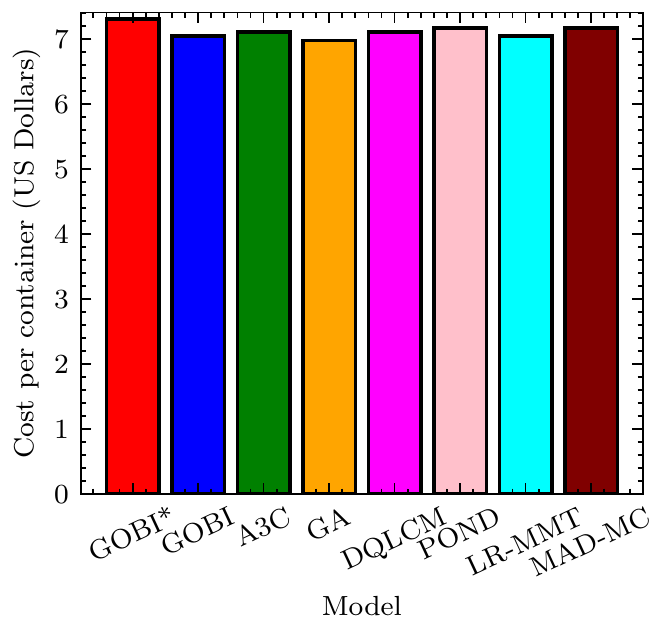}
    \label{fig:f_cost}
    }
    \subfigure[Fairness per application]{
    \includegraphics[width=.23\textwidth]{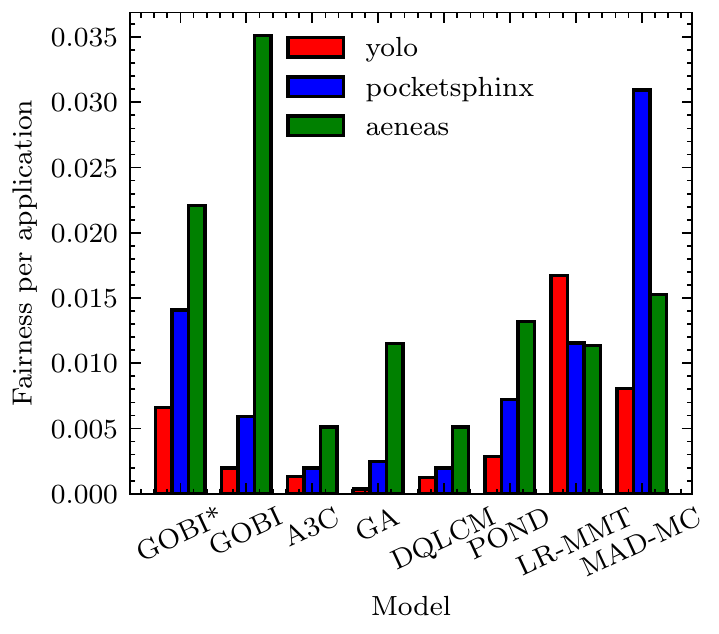}
    \label{fig:f_fairness_pa}
    }\\
    \subfigure[Average number of migrations]{
    \includegraphics[width=.23\textwidth]{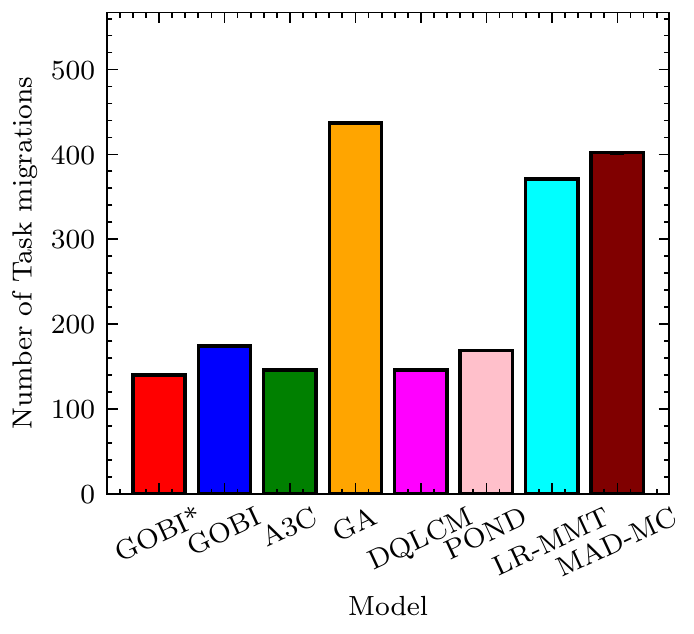}
    \label{fig:f_migrations}
    }
    \subfigure[Average Migration Time with execution time]{
    \includegraphics[width=.23\textwidth]{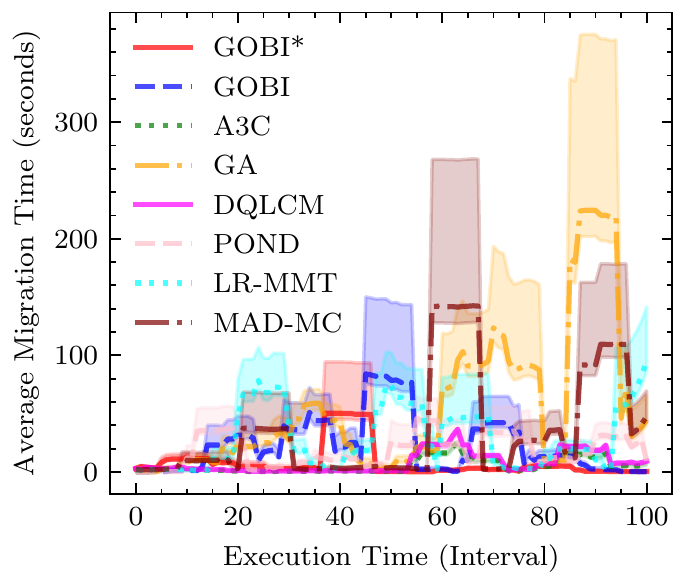}
    \label{fig:f_migrations_time}
    }
    \subfigure[Average Wait Time with execution time]{
    \includegraphics[width=.23\textwidth]{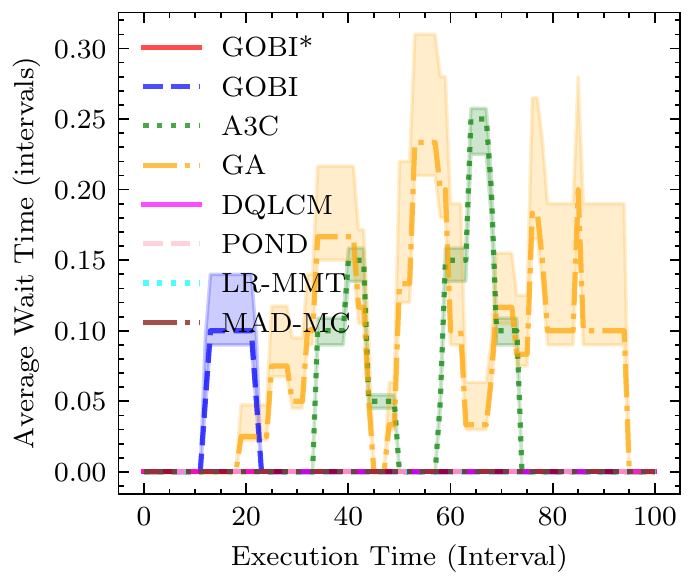}
    \label{fig:f_wait_time}
    }
    \subfigure[Average Response Time with execution time]{
    \includegraphics[width=.23\textwidth]{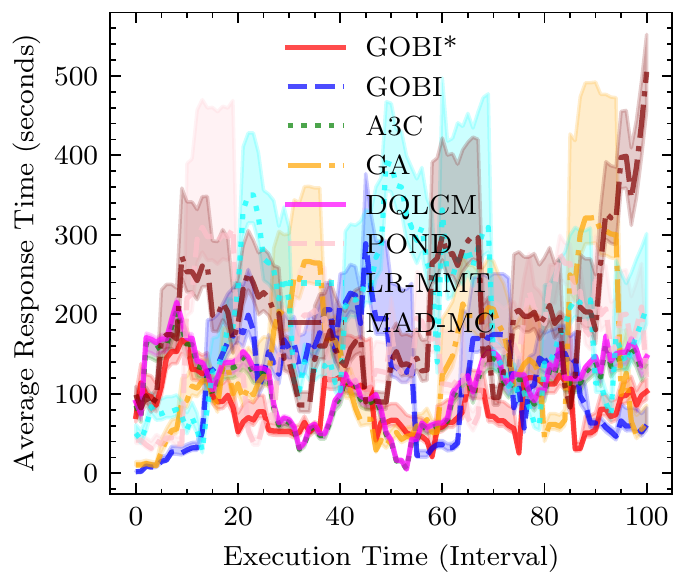}
    \label{fig:f_response_time}
    }
    \caption{Comparison of GOBI and GOBI* against baselines on framework with 10 hosts}
    \label{fig:framework_results2}
\end{figure*}

\begin{figure*}
    \centering
    \subfigure[Average CPU Utilization of hosts]{
    \includegraphics[width=.23\textwidth]{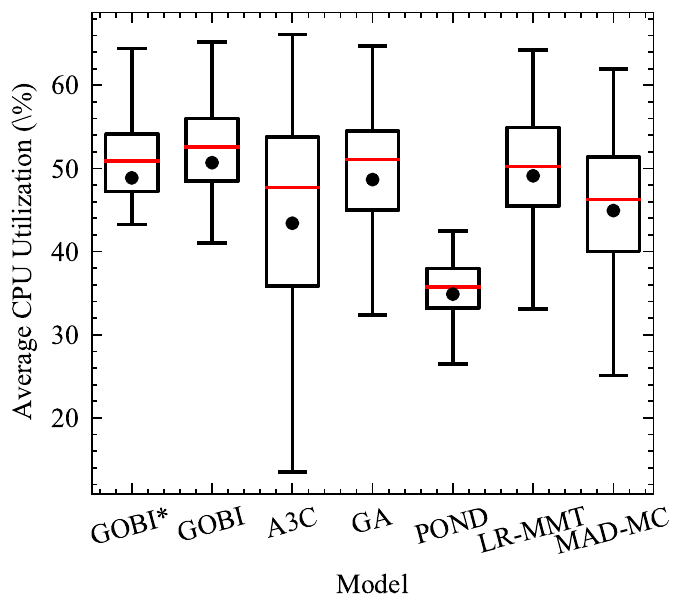}
    \label{fig:s_cpu}
    }
    \subfigure[Average Number of active containers with simulation time]{
    \includegraphics[width=.23\textwidth]{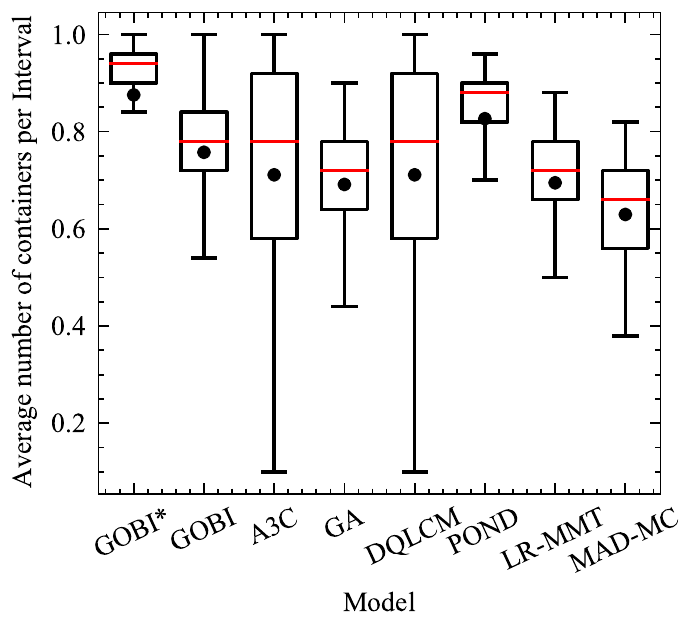}
    \label{fig:s_containers}
    }
    \subfigure[Average Cost per container (US Dollars)]{
    \includegraphics[width=.23\textwidth]{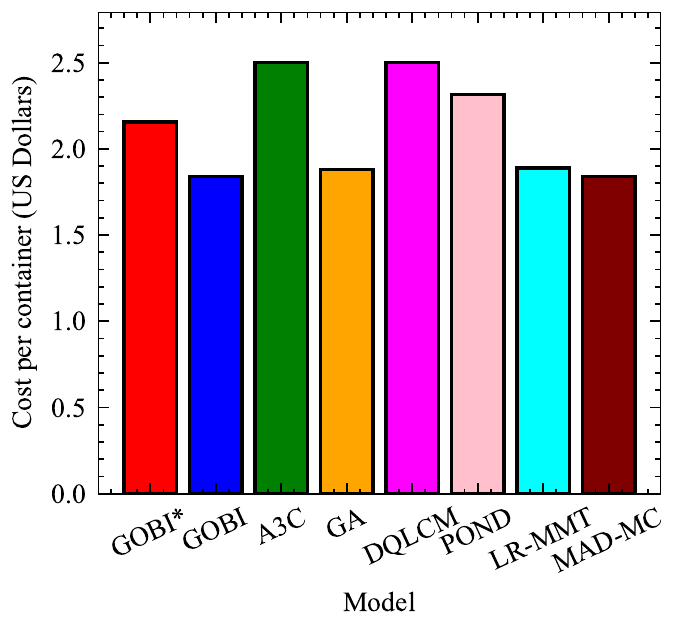}
    \label{fig:s_cost}
    }
    \subfigure[Average Scheduling Time]{
    \includegraphics[width=.23\textwidth]{images/simulator/Bar-Scheduling_Time__seconds_.pdf}
    \label{fig:s_scheduling}
    }\\
    \subfigure[Average number of migrations]{
    \includegraphics[width=.23\textwidth]{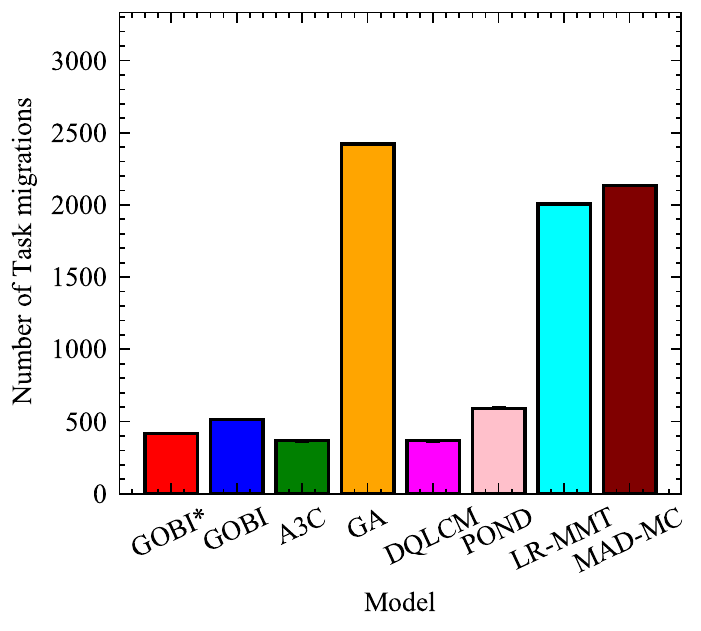}
    \label{fig:s_migrations}
    }
    \subfigure[Average Migration Time with simulation time]{
    \includegraphics[width=.23\textwidth]{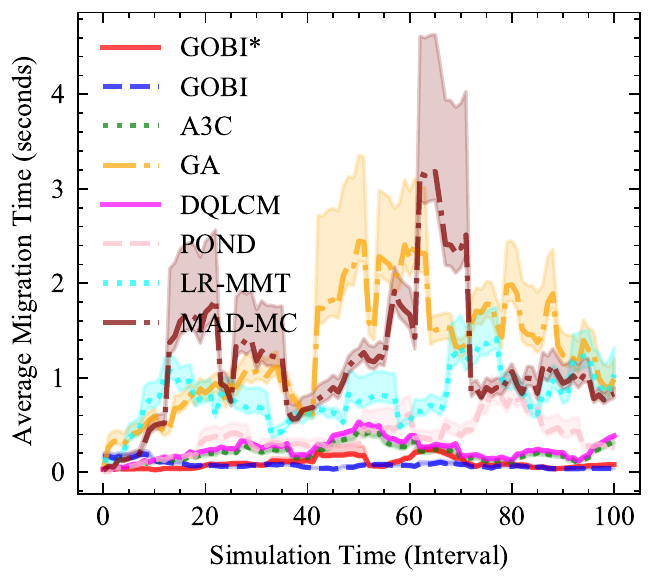}
    \label{fig:s_migrations_time}
    }
    \subfigure[Average Wait Time with simulation time]{
    \includegraphics[width=.23\textwidth]{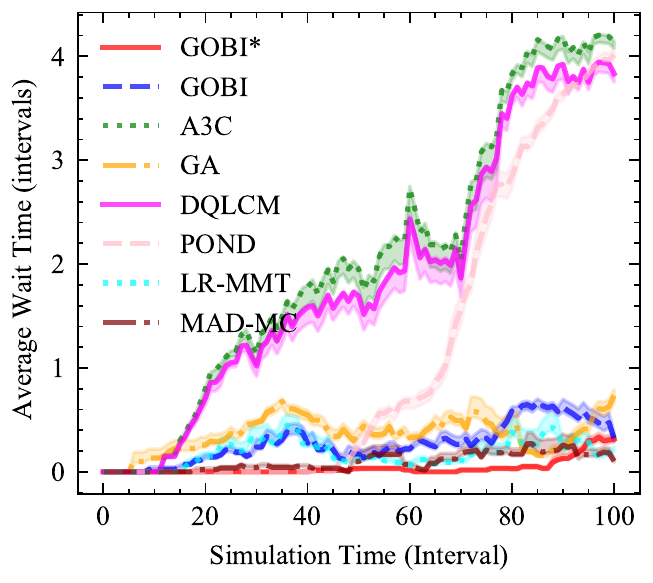}
    \label{fig:s_wait_time}
    }
    \subfigure[Average Response Time with simulation time]{
    \includegraphics[width=.23\textwidth]{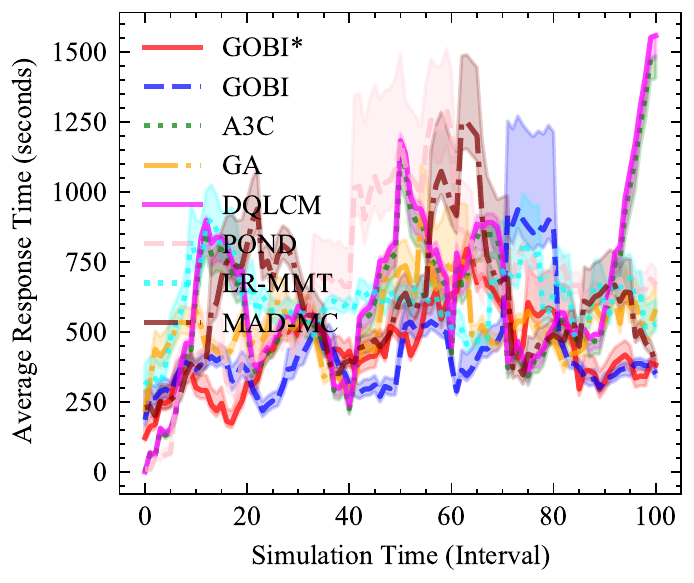}
    \label{fig:s_response_time}
    }
    \caption{Comparison of GOBI and GOBI* against baselines on simulator with 50 hosts}
    \label{fig:simulator_results2}
\end{figure*}

\subsection{Model Training}
To train the neural approximator $f$ of the GOBI model, we first form a dataset by running a random scheduler which gives us data of the form $\Lambda = \{ [\phi(A_{t-1}), \phi(H^{t-1}), \phi(\mathcal{D})], \mathcal{O}(P_t) \}_t$. The random scheduler selects a random number of tasks at any interval $I_t$ and decides to migrate them to any host uniformly at random. Using this, we train our model $f(x; \theta)$ such that $\theta = \argmin_{\hat{\theta}}{\sum_{(x,\phi(\mathcal{D}))\in\Lambda} [ \mathcal{L}(f(x;\hat{\theta}), \phi(\mathcal{D}))]}$. We split the dataset into $80\%$ as training data and the rest $20\%$ for testing. We use Mean Square Error (MSE) as the loss function with AdamW optimizer. We use learning rate of $10^{-3}$ and weight decay of $10^{-5}$. The size of $\Lambda$ dataset to train the GOBI model was $2000$ for both cases of $10$ and $50$ hosts.

To train the neural approximator $f^*$ of the GOBI* model, we first train LSTM models to predict the utilization metrics in $I_t$ using the utilization values of the previous intervals. Using the decision output from the GOBI model and the predicted utilization metrics, we simulate and get an estimate of $I_t$'s QoS parameters. Adding these in the dataset, we now have $\Lambda^* = \{ [\phi(A_{t-1}), \phi(H^{t-1}), \mathcal{O}(\bar{P}_t), \phi(\mathcal{D})] \}$, which we use to train $f^*(x, \theta^*)$. Figures~\ref{fig:simulator_training} and \ref{fig:framework_training} show the loss on the test set with epochs for different models.

As dataset generation from event-driven simulator is much faster than actual execution on physical systems, to train models for the framework, we first pre-train using a dataset of size $2000$ (generated using simulations) and then fine-tune using a smaller dataset generated by running on physical setup of size $200$. We train on the dataset generated using the simulator for $600$ epochs and till convergence on the dataset generated on the physical setup. This accounts for the sudden change in the loss values seen in Figure~\ref{fig:framework_training} at the $600$-th epoch.

\section{Power Models of Azure Hosts}

The power consumption values of increments of 10\% CPU utilization of Azure virtual machine types are taken from Standard Performance Evaluation Corporation (SPEC) benchmark repository \texttt{https://www.spec.org/cloud\_iaas2018/results/} and are shown in Table~\ref{tab:host_powers}. For intermediate CPU utilization values we use linear interpolation to calculate power consumption.

\section{Additional Results}

Figures~\ref{fig:framework_results2} and \ref{fig:simulator_results2} show additional results on the azure testbed (with 10 hosts) and simulation environment (with 50 hosts). We also show the \textit{Average Cost per container} given as \[\frac{\sum_t \sum_{h_i\in H} \int_{t = t_{i}}^{t_{i+1}} C_{h_i}(t)dt}{\sum_t |L_t|},\]where $C_{h_i}(t)$ is the cost function of host $h_i$ at instant $t$.

Figures~\ref{fig:f_cpu} and \ref{fig:s_cpu} show the mean CPU utilization of host machines averaged across the whole execution duration. We see that the CPU utilization is minimum for the \textit{POND} baseline with average at $8.94\%$. Both \textit{GOBI} and \textit{GOBI*} have a higher average CPU utilization by $2.78\%-6.12\%$. Figures~\ref{fig:f_containers} and \ref{fig:s_containers} show the average number of active containers for the various scheduling approaches. GOBI and GOBI*, in both cases, have the highest number of active containers. This is due to the efficient placement of these two approaches capable of executing multiple containers together. Figures~\ref{fig:f_cost} and \ref{fig:s_cost} show the average execution cost of all policies. For 10 hosts, all policies have negligible difference in the average cost per container. For 50 hosts, GOBI* and GOBI have $24.7\%$ and $26.3\%$ lower cost compared to the \textit{A3C} baseline, with GOBI having the least average cost among all policies.

\begin{figure}
    \centering
    \subfigure[CPU utilization (\%) of active containers with time]{
    \includegraphics[width=0.95\columnwidth]{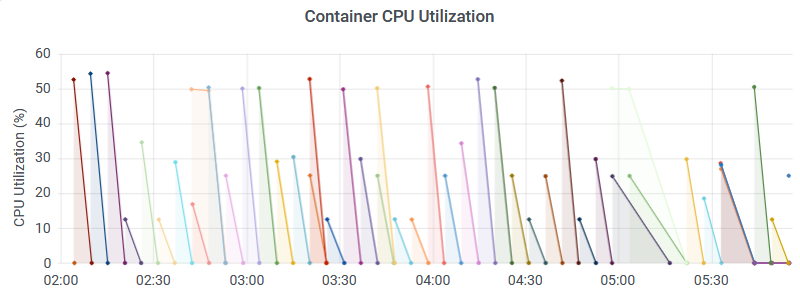}
    \label{fig:grafana_1}
    }
    \subfigure[CPU utilization (\%) of hosts with time]{
    \includegraphics[width=.95\columnwidth]{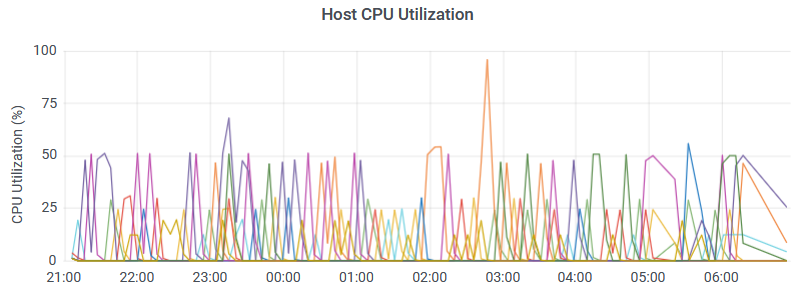}
    \label{fig:grafana_2}
    }
    \caption{Visualizing utilization metrics using the Grafana toolkit.}
    \label{fig:grafana}
\end{figure}

Figure~\ref{fig:f_fairness_pa} show the fairness metric for each application for all policies. The \textit{A3C} baseline is the least fair among all policies with the Jain's indices for each application being in the range $0.01-0.05$. Some other observations include \textit{MAD-MC} being fair to the Pocketsphinx application and \textit{GOBI} being fair to the Aeneas application. Average fairness index across applications is highest for the \textit{GOBI*} approach being $0.034$. Figures~\ref{fig:f_migrations} and \ref{fig:s_migrations} show the average number of container migrations for each policy. GOBI and GOBI* have consistently low number of migrations ($140-174$). GOBI* has the least average migration count in the case of 10 hosts ($146$). However, the \textit{A3C} baseline has the least migrations in the case of 50 hosts ($486$). Figures~\ref{fig:f_migrations_time} and \ref{fig:s_migrations_time} show the variation of average migration time with time. Migration times are lowest for \textit{GOBI}, \textit{GOBI*} and heuristic based methods. Figures~\ref{fig:f_wait_time} and \ref{fig:s_wait_time} show the variation of average wait time for tasks with time. Wait times are fairly stable for all policies except \textit{POND}, \textit{A3C} and \textit{DQLCM}. This is because \textit{POND} fails to converge to solutions frequently after 60-th interval. The reinforcement learning approaches, \textit{A3C} and \textit{DQLCM} are unable to quickly adapt to volatile workloads in relatively larger scale setup and hence tend to converge to infeasible migration decision (like migrating tasks to a host running at capacity). This leads to gradually increasing wait times for tasks in the simulated environment with 50 nodes. These effects are not as pronounced in the physical setup with only 10 hosts (Figure~\ref{fig:f_wait_time}), Figures~\ref{fig:f_response_time} and \ref{fig:s_response_time} show the variation of average response time for tasks with time. Clearly, the response times are highly volatile due to the non-stationary nature of the workloads. Overall, \textit{GOBI} and \textit{GOBI*} have the least overall response times among all policies.

\section{Visualization of Utilization Metrics}

We used the Grafana toolkit~\cite{grafana} to visualize the utilization metrics and check the volatility of the workloads running in the fog environment. These graphs are dynamically generated from the data in the InfluxDB~\cite{naqvi2017time} database. Screenshots of a few workload visualizations are shown in Figure~\ref{fig:grafana}. The figure demonstrates the significant volatility in the resource consumption of containers across a few hours.

\bibliographystyle{IEEEtran}
\bibliography{ref}

\begin{IEEEbiography}
[{\includegraphics[width=1in,height=1in,clip,keepaspectratio]{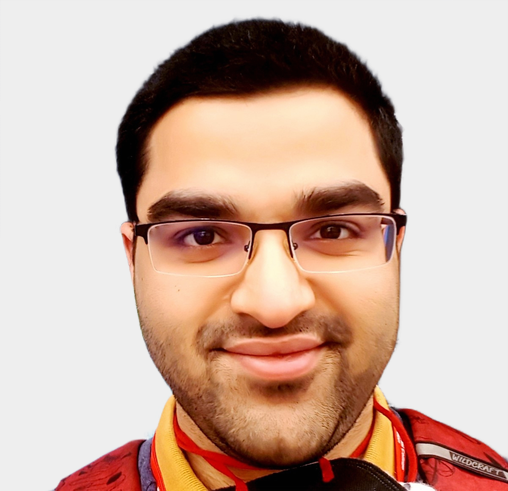}}]
{Shreshth Tuli}
is a President's Ph.D. Scholar at the Department of Computing, Imperial College London, UK. Prior to this he was an undergraduate student at the Department of Computer Science and Engineering at Indian Institute of Technology - Delhi, India. He has worked as a visiting research fellow at the CLOUDS Laboratory, School of Computing and Information Systems, the University of Melbourne, Australia. He is a national level Kishore Vaigyanik Protsahan Yojana (KVPY) scholarship holder from the Government of India for excellence in science and innovation. His research interests include Internet of Things (IoT), Fog Computing and Deep Learning.
\end{IEEEbiography}
\begin{IEEEbiography}
[{\includegraphics[width=1in,height=1in,clip,keepaspectratio]{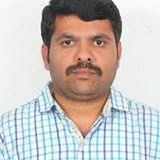}}]
{Shivananda R Poojara}
is a Ph.D student at Mobile and Cloud Computing Laboratory, University of Tartu, Estonia. He is also working as a Junior Research Fellow in IT Academy program. His research interests are serverless computing, edge analytics, fog and cloud computing. He is a member of IEEE, IET and ISTE. He has a 6+ years of teaching experience and worked at Nokia R\&D Labs as an intern.
\end{IEEEbiography}
\begin{IEEEbiography}
[{\includegraphics[width=1in,height=1in,clip,keepaspectratio]{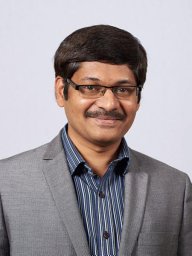}}]
{Satish N. Srirama}
is an Associate Professor at School of Computer and Information Sciences, University of Hyderabad,  India. He is also a Visiting Professor and the honorary head of the Mobile \& Cloud Lab at the Institute of Computer Science, University of Tartu, Estonia. His current research focuses on cloud computing, mobile cloud, IoT, Fog computing, migrating scientific computing and large scale data analytics to the cloud. He received his PhD in computer science from RWTH Aachen University in 2008. He is an IEEE Senior Member and an Editor of Wiley Software: Practice and Experience journal.
\end{IEEEbiography}
\vspace{-2.5in}
\begin{IEEEbiography}
[{\includegraphics[width=1in,height=1in,clip,keepaspectratio]{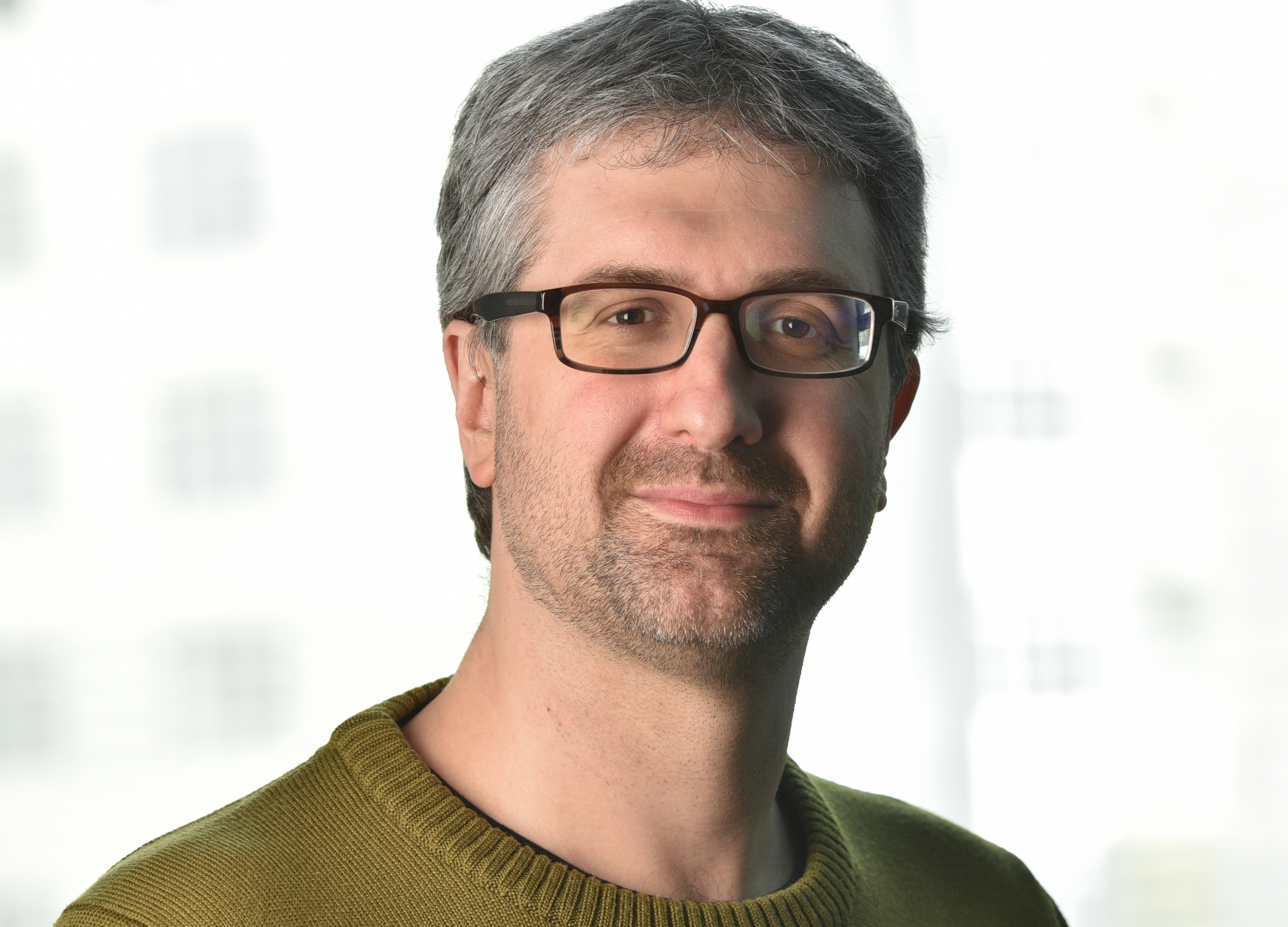}}]
{Giuliano Casale}
joined the Department of Computing at Imperial College London in 2010, where he is currently a Reader. Previously, he worked as a research scientist and consultant in the capacity planning industry. He teaches and does research in performance engineering and cloud computing, topics on which he has published more than 100 refereed papers. He has served on the technical program committee of over 80 conferences and workshops and as co-chair for several conferences in the area of performance and reliability engineering, such as ACM SIGMETRICS/Performance and IEEE/IFIP DSN. His research work has received multiple awards, recently the best paper award at ACM SIGMETRICS. He serves on the editorial boards of IEEE TNSM and ACM TOMPECS and as current chair of ACM SIGMETRICS.
\end{IEEEbiography}
\vspace{-2.45in}
\begin{IEEEbiography}
[{\includegraphics[width=1in,height=1in,clip,keepaspectratio]{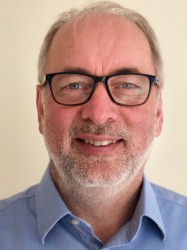}}]
{Nicholas R. Jennings}
is the Vice-Provost for Research and Enterprise and Professor of Artificial Intelligence at Imperial College London. He is an internationally-recognised authority in the areas of AI, autonomous systems, cyber-security and agent-based computing. He is a member of the UK government’s AI Council, the governing body of the Engineering and Physical Sciences Research Council, the Monaco Digital Advisory Council,  and chair of the Royal Academy of Engineering’s Policy Committee.   Before Imperial, he was the UK's first Regius Professor of Computer Science (a post bestowed by the monarch to recognise exceptionally high quality research) and the UK Government’s first Chief Scientific Advisor for National Security.
\end{IEEEbiography}
\end{document}